\documentclass[11pt]{article}
\usepackage{fullpage}
\usepackage{url}
\usepackage{xspace}

\usepackage{graphics}
\usepackage[dvips]{epsfig}

\usepackage{amsmath}
\usepackage{amssymb}
\usepackage{amsfonts}
\usepackage{graphicx}
\usepackage{algorithm}
\usepackage{comment}
\usepackage{multirow}
\usepackage{algpseudocode}

\newtheorem{theorem}{Theorem}[section]
\newtheorem{lemma}{Lemma}[section]
\newtheorem{corollary}{Corollary}[section]
\newtheorem{claim}{Claim}[section]
\newtheorem{proposition}{Proposition}[section]

\newtheorem{fact}{Fact}[section]

\newcommand{\qed}{\hfill $\Box$ \bigbreak}
\newenvironment{proof}{\noindent {\bf Proof.}}{\qed}

\newcommand{\remove}[1]{}

\newcommand{\diam}[1]{\ensuremath{{diam}(#1)}}

\begin{document}

\baselineskip  0.2in 
\parskip     0.1in 
\parindent   0.0in 

\title{{\bf Time vs. Information Tradeoffs\\ for Leader Election in Anonymous Trees}}
\date{}
\newcommand{\inst}[1]{$^{#1}$}

\author{
Christian Glacet\inst{1}$^,$\footnote{Partially supported by the ANR project DISPLEXITY (ANR-11BS02-014)},
Avery Miller\inst{2},
Andrzej Pelc\inst{2}$^,$\footnote{Partially supported by NSERC discovery grant and by the Research Chair in Distributed Computing at the Universit\'e du Qu\'{e}bec en Outaouais.}\\
\inst{1}CNR -- IEIIT, Torino, Italy\\
\inst{2} Universit\'{e} du Qu\'{e}bec en Outaouais, Gatineau, Canada.\\
E-mails: \url{christian.glacet@gmail.com}, \url{avery@averymiller.ca}, \url{pelc@uqo.ca}\\
}

\date{ }
\maketitle

\begin{abstract}
Leader election is one of the fundamental problems in distributed computing. It calls for all nodes of a network to agree on a single node, called the leader.
If the nodes of the network have distinct labels, then agreeing on a single node means that all nodes have to output the label of the elected leader.
 If the nodes of the network are anonymous,
the task of leader election is formulated as follows: every node $v$ of the network must output a simple path, which is coded as a sequence of port numbers, such that
all these paths 
end at a common node, the leader. In this paper, we study deterministic leader election in anonymous trees. 

Our aim is to establish tradeoffs between the allocated time $\tau$ and the amount of information that has to be given {\em a priori} to the nodes to enable leader election in time $\tau$ in all trees for which leader election in this time is at all possible.
Following the framework of {\em algorithms
with advice}, this information (a single binary string) is provided to all nodes at the start by an oracle knowing the entire tree. The length of this string is called the {\em size of advice}.  
 For a given time $\tau$ allocated to leader election, we give upper and lower bounds on the minimum size
of advice sufficient to perform leader election in time $\tau$.  
 
 For most values of $\tau$, our upper and lower bounds are either tight up to multiplicative constants, or
they differ only by a logarithmic factor.  Let $T$ be an $n$-node tree of diameter $diam \leq D$.  
While leader election in time $diam$ can be performed without any advice, for time $diam-1$ we give
tight upper and lower bounds of $\Theta (\log D)$. For time $diam-2$ we give
tight upper and lower bounds of $\Theta (\log D)$ for even values of $diam$, and
tight upper and lower bounds of $\Theta (\log n)$ for odd values of $diam$.
Moving to shorter time, in the interval $[\beta \cdot diam, diam -3]$ for constant $\beta >1/2$, we prove an upper bound of $O(\frac{n\log n}{D})$ and a lower bound
of $\Omega(\frac{n}{D})$, the latter being valid whenever $diam$ is odd or when the time is at most $diam-4$. Hence, with the exception of the special case
when $diam$ is even and time is exactly $diam-3$, our bounds leave only a logarithmic gap in this time interval. Finally, for time $\alpha \cdot diam$ for any constant $\alpha <1/2$ (except for the case of very small diameters),  we again give tight upper and lower bounds, this time $\Theta (n)$.

\vspace{2ex}

\end{abstract}

\vfill

\vfill

\thispagestyle{empty}
\setcounter{page}{0}
\pagebreak

\section{Introduction}

{\bf Background.} 
Leader election is one of the fundamental problems in distributed computing \cite{Ly}. It calls for all nodes of a network to agree on a single node, called the leader.
This classic problem was first formulated in \cite{LL} in the study of local area token ring networks, where, at all times, exactly one node (the owner of a circulating token) has the right to initiate
communication. When the token is accidentally lost, a leader is elected as the initial owner of the token.

If the nodes of the network have distinct labels, then agreeing on a single node means that all nodes have to output the label of the elected leader. However, in many
applications, even if nodes have distinct identities, they may  be reluctant to reveal them, e.g., for privacy or security reasons. Hence it is important to design leader election algorithms that do not depend on the knowledge of such labels and that can work in anonymous networks as well. Under this scenario, agreeing on a single node means
that every node has to output a simple path (coded as a sequence of port numbers) to a common node.

\noindent
{\bf Model and Problem Description.} The network is modeled as an undirected connected graph with $n$ nodes and with diameter $diam$ at most $D$.
In this paper, we restrict attention to tree networks, i.e., connected networks without cycles.
We denote by $\diam{T}$ the diameter of tree $T$. Nodes do not have any identifiers.
On the other hand, we assume that, at each node $v$,
each edge incident to $v$ has a distinct {\em port number} from 
$\{0,\dots,d-1\}$, where $d$ is the degree of $v$. Hence each edge has two corresponding port numbers, one at each of its endpoints. 
Port numbering is {\em local} to each node, i.e., there is no relation between
port numbers at  the two endpoints of an edge. Initially each node knows only its own degree.
The task of leader election is formulated as follows. Every node $v$ of the tree must output a sequence $P(v)=(p_1,\dots,p_k)$ of nonnegative integers.
For each node $v$, let $P^*(v)$ be the path starting at $v$ that results from taking the number $p_i$ from $P(v)$ as the outgoing port at the $i^{th}$ node of the path.
All paths $P^*(v)$ must be simple paths in the tree that end at a common node, called the leader. In this paper, we consider deterministic leader election algorithms.
 
Note that, in the absence of port numbers, there would be no way to identify the elected leader by non-leaders, as all
ports, and hence all neighbours, would be indistinguishable to a node.
Security and privacy reasons for not revealing node identifiers are irrelevant in the case of port numbers. 

We use the extensively studied $\cal{LOCAL}$ communication model \cite{Pe}. In this model, communication proceeds in synchronous
rounds and all nodes start simultaneously. In each round, each node
can exchange arbitrary messages with all of its neighbours and perform arbitrary local computations. For any tree $T$,  any $r \geq 0$ and any node $x$ in $T$, we use $V_T(x,r)$ to denote     
the {\em view} acquired in $T$ by $x$ within $r$ communication rounds. This is all the information that $v$ gets about the tree $T$ in $r$ rounds. Thus, the view $V_T(x,r)$ in $T$  consists of the subtree of $T$ induced by all nodes at distance at most $r$
from $x$, together with  all the port numbers at these nodes, and with the degrees of all nodes at distance exactly $r$ from $x$. 
If no additional knowledge is provided {\em a priori} to the nodes, the decisions of $x$ in round $r$ in any deterministic algorithm are a function of $V_T(x,r)$. 
In most cases, a node's view is considered in the underlying tree in which leader election is being solved,
and then the subscript $T$ is omitted.
The {\em time} of leader election is the minimum number of rounds sufficient to complete it by all nodes. 

It is well known that the synchronous process of the $\cal{LOCAL}$  model can be simulated in an asynchronous network. This can be achieved 
by defining for each node separately its asynchronous round $i$;
in this round, a node performs local computations, then sends messages stamped $i$ to all neighbours, and  waits until it gets messages stamped $i$ from all neighbours.
To make this work, every node is required to send at least one (possibly empty) message with each stamp, until termination.
All of our results can be translated for asynchronous networks by replacing ``time of completing a task''  by ``the maximum number of asynchronous rounds  to complete it, taken over all nodes''. 

For anonymous trees, the task of leader election is not always feasible, regardless of the allocated time. This is the case when the tree is {\em symmetric},
i.e., when there exists a non-trivial port-preserving automorphism of it. Such an automorphism is defined as a bijection $f : X \rightarrow X$, where $X$ is the set of  nodes, such that
$\{x, y\}$ is an edge with port numbers $p$ at $x$ and $q$ at $y$ if and only if $\{f(x), f(y)\}$ is an edge with port numbers $p$ at $f(x)$ and $q$ at $f(y)$.
It is easy to see that leader election is possible in a tree only if the tree is not symmetric. Symmetric trees are easy to characterize.
Indeed, every tree has a {\em centre} which is either a node or an edge defined as follows. If the diameter $diam$ is even, then the {\em central node} is the unique node in the middle
of every simple path of length $diam$, and if the diameter $diam$ is odd, then the {\em central edge} is the unique edge in the middle
of every simple path of length $diam$. A tree is symmetric if and only if $diam$ is odd, ports at the central edge are equal, and the two subtrees resulting from the deletion of the central edge are (port-preserving) isomorphic. For symmetric trees, the only non-trivial automorphism is the one switching the corresponding nodes of these subtrees, and this prevents leader election. 

Moreover, even in non-symmetric trees, leader election may be impossible if the allocated time is too short. Consider the line of length 6 with port numbers
0,0,1,1,0,0,1,1,0,1,0,0 (from left to right). If the allocated time is 1, then leader election is impossible even if nodes know {\em a priori} the entire map of the line.
Indeed, neither of the two leaves knows whether it is the left or the right leaf and cannot learn this fact in time 1, and thus, leaves cannot output correct simple paths to a common node
(the formal proof is slightly more complicated). 
Hence, for any non-symmetric tree $T$, it is important to introduce the parameter $\xi(T)$ defined as the minimum time in which leader election is feasible, assuming that each node is given the entire map of $T$ with all port numbers faithfully mapped (but without the position of the node marked in the map). For the line $T$ in the above example, $\xi(T)=2$.

Our aim is to establish tradeoffs between the allocated time and the amount of information that has to be given {\em a priori} to the nodes to enable them to perform
leader election.
Following the framework of {\em algorithms
with advice}, see, e.g.,   \cite{DP,EFKR,FGIP,FKL,FP,IKP,SN}, this information (a single binary string) is provided to all nodes at the start by an oracle knowing the entire tree. The length of this string is called the {\em size of advice}. Of course, since the faithful map of the tree is the total information about it, asking about the minimum size of advice
to solve leader election in time $\tau$ is meaningful only in the class of trees $T$ for which $\xi(T) \leq \tau$, because otherwise, no advice can help. In light of these remarks, we are able to precisely formulate
the central problem of this paper.
\begin{quotation}
\noindent
For a given time $\tau$, what is the minimum size of advice that permits leader election in time $\tau$ for all trees $T$ where $\xi(T) \leq \tau$?
\end{quotation}  

The paradigm of algorithms with advice has a far-reaching significance in the domain of network algorithms. Establishing a tight bound on the minimum size of advice sufficient to accomplish a given task permits to rule out
entire classes of algorithms and thus focus only on possible candidates. For example, if we prove that $\Theta(\log n)$ bits of advice are needed to perform a certain task in $n$-node trees, this rules out all 
potential algorithms that can work using only the diameter $diam$ of the tree, 
as $diam$ can be given
to the nodes using $\Theta(\log (diam))$ bits, and the diameter can be, e.g.,  logarithmic in the size of the tree. Lower bounds on the size of advice
give us impossibility results based strictly on the \emph{amount} of initial knowledge outlined in a model's description.
This more general approach should be contrasted with
traditional results that focus on specific \emph{kinds} of information available to nodes, such as the size, diameter, or maximum node degree.

\noindent
{\bf Our results.} 
Let $T$ be an $n$-node tree of diameter $diam \leq D$. For a given time $\tau$ allocated to leader election, we give upper and lower bounds on the minimum size
of advice sufficient to perform leader election in time $\tau$. An upper bound $U$ means that, for all trees $T$ with $\xi(T) \leq \tau$, leader election in time $\tau$ is possible given advice of size $O(U)$. We prove such a bound by constructing advice of size $O(U)$ together with a leader election algorithm for all trees $T$ with $\xi(T) \leq \tau$ that uses this advice and works in time $\tau$.
 A lower bound $L$ means that there exist trees $T$ with $\xi(T) \leq \tau$ for which leader election in time $\tau$  requires advice of size $\Omega(L)$.
 Proving such a bound means constructing a class consisting of trees $T$ with $\xi(T) \leq \tau$ for which no leader election algorithm running in time $\tau$ with advice of size $o(L)$
 can succeed.

For most values of $\tau$, our upper and lower bounds are either tight up to multiplicative constants, or
they differ only by a logarithmic factor. More precisely, these bounds are the following. 
While leader election in time $diam$ can be performed without any advice, for time $diam-1$ we give
tight upper and lower bounds of $\Theta (\log D)$. For time $diam-2$, we give
tight upper and lower bounds of $\Theta (\log D)$ for even values of $diam$ and
tight upper and lower bounds of $\Theta (\log n)$ for odd values of  $diam$.
Moving to shorter time, in the interval $[\beta \cdot diam, diam -3]$ for constant $\beta >1/2$, we prove an upper bound of $O(\frac{n\log n}{D})$ and a lower bound
of $\Omega(\frac{n}{D})$, the latter valid whenever $diam$ is odd or time is at most $diam-4$. Hence, with the exception of the special case
when $diam$ is even and time is exactly $diam-3$, our bounds leave only a logarithmic gap in this time interval. (See section 7 for a discussion of this special case.)
Finally, for time $\alpha \cdot diam$ for any constant $\alpha <1/2$ (except for the case of very small diameters, namely for $diam \in \omega (\log ^2n)$)  we again give tight upper and lower bounds, this time $\Theta (n)$. 
The above results are summarized in Figure \ref{results}.

\begin{figure}
\begin{center}
\[
\begin{array}{|c|l|}
\hline
\textbf{Time} & \parbox[m]{9cm}{\centering \textbf{Minimum size of advice}}\\

\hline
diam & 0 \\
\hline
\parbox[m]{5cm}{\centering $diam-1$} & \Theta(\log{D})\\
\hline
\multirow{2}{*}{$diam-2$} & \Theta(\log{D}) \textrm{ for even $diam$}\\
 & \Theta(\log{n}) \textrm{ for odd $diam$}\\
\hline
\multirow{3}{*}{\parbox[t][][t]{5cm}{\centering $\beta \cdot diam \leq \textrm{Time} \leq diam-3$\\ $\textrm{for constant } \beta > \frac{1}{2}$ }} & O(\frac{n\log{n}}{D}) \textrm{ upper bound} \\
 & \Omega(\frac{n}{D}) \textrm{ lower bound for odd $diam$ or Time $\leq diam-4$}\\
 & \textrm{? lower bound for even $diam$ and Time $= diam-3$}\\
\hline
\multirow{3}{*}{\parbox[t][][t]{5cm}{\centering $\alpha \cdot diam$ \\ \textrm{for constant $\alpha < \frac{1}{2}$ and $diam \in \omega(\log^2{n})$}}} & \ \\
 
 & \Theta(n)\\
 & \ \\

\hline
\end{array}
\]
\end{center}
\caption{Tradeoffs between time and size of advice in $n$-node trees with diameter $diam \leq D$}
\label{results}
\end{figure}

Our results show that the minimum size of advice sufficient to perform leader election is very sensitive to the amount of time allocated to this task, and that
this sensitivity occurs at different time values depending on the relation between the diameter and the size of the tree. If $diam$ is odd and small
compared to $n$, e.g., $diam \in O(\log n)$, then a difference of one round (between $diam -1$ and $diam -2$) causes an exponential jump of the size of information
required for leader election, and {\em another} exponential jump occurs in this case between time $diam -2$ and $diam -3$. By contrast, for larger diameter, e.g.,
$diam \in \Theta (\sqrt{n})$, the first exponential jump disappears but the second still holds. On the other hand, perhaps surprisingly, an exponential jump
occurs at fixed time $diam -2$ when the diameter is small (e.g., logarithmic in $n$), depending only on the parity of the diameter.

\noindent
{\bf Related work.}
The leader election problem was introduced in \cite{LL}. This problem  was first extensively studied in the scenario 
where all nodes have distinct labels. Initially, it was investigated for rings.
A synchronous algorithm based on label comparisons and using
$O(n \log n)$ messages was given in \cite{HS}.  In \cite{FL} it was proved that
this complexity is optimal for comparison-based algorithms. On the other hand, the authors showed
an algorithm using a linear number of messages but requiring very large running time.
An asynchronous algorithm using $O(n \log n)$ messages was given, e.g., in \cite{P}, and
the optimality of this message complexity was shown in \cite{B}. Deterministic leader election in radio networks has been studied, e.g., 
in \cite{JKZ,KP,NO}, as well as randomized leader election, e.g., in \cite{Wil}. In \cite{HKMMJ}, the leader election problem was
approached in a model based on mobile agents for networks with labeled nodes.

Many authors \cite{An,AtSn,ASW,BSVCGS,BV,YK2,YK3} studied leader election
in anonymous networks. In particular, \cite{BSVCGS,YK3} characterize message-passing networks in which
leader election can be achieved when nodes are anonymous. In \cite{YK2}, the authors study
the problem of leader election in general networks under the assumption that node labels are
not unique. They characterize networks in which this can be done and give an algorithm
which performs election when it is feasible. 
In  \cite{FKKLS},  the authors
study feasibility and message complexity of leader election in rings with possibly
nonunique labels, while, in \cite{DoPe}, the authors provide algorithms for a
generalized leader election problem in rings with arbitrary labels,
unknown (and arbitrary) size of the ring, and for both
synchronous and asynchronous communication. 
Memory needed for leader election in unlabeled networks was studied in \cite{FP}. 
In \cite{FP1}, the authors investigated the time of leader election in anonymous networks
by characterizing this time in terms of the network size, the diameter of the network, and an additional
parameter called the level of symmetry, which measures how deeply nodes have to inspect the network in order to notice differences in their views of it.
In \cite{DP1}, the authors studied the feasibility of leader election among anonymous agents that
navigate in a network in an asynchronous way.

Providing nodes or agents with arbitrary kinds of information that can be used to perform network tasks more efficiently has previously been
proposed in \cite{AKM01,DP,EFKR,FGIP,FIP1,FIP2,FKL,FP,FPR,GPPR02,IKP,KKKP02,KKP05,MP,SN,TZ05}. This approach was referred to as
{\em algorithms with advice}.  
The advice is given either to the nodes of the network or to mobile agents performing some network task.
In the first case, instead of advice, the term {\em informative labeling schemes} is sometimes used if (unlike in our scenario) different nodes can get different information.

Several authors studied the minimum size of advice required to solve
network problems in an efficient way. 
 In \cite{KKP05}, given a distributed representation of a solution for a problem,
the authors investigated the number of bits of communication needed to verify the legality of the represented solution.
In \cite{FIP1}, the authors compared the minimum size of advice required to
solve two information dissemination problems using a linear number of messages. 
In \cite{FKL}, it was shown that advice of constant size given to the nodes enables the distributed construction of a minimum
spanning tree in logarithmic time. 
In \cite{EFKR}, the advice paradigm was used for online problems.
In \cite{FGIP}, the authors established lower bounds on the size of advice 
needed to beat time $\Theta(\log^*n)$
for 3-coloring cycles and to achieve time $\Theta(\log^*n)$ for 3-coloring unoriented trees.  
In the case of \cite{SN}, the issue was not efficiency but feasibility: it
was shown that $\Theta(n\log n)$ is the minimum size of advice
required to perform monotone connected graph clearing.
In \cite{IKP}, the authors studied radio networks for
which it is possible to perform centralized broadcasting in constant time. They proved that constant time is achievable with
$O(n)$ bits of advice in such networks, while
$o(n)$ bits are not enough. In \cite{FPR}, the authors studied the problem of topology recognition with advice given to the nodes. 
In \cite{DP}, the task of drawing an isomorphic map by an agent in a graph was considered, and the problem was to determine the minimum advice that has to be given to the agent
for the task to be feasible. 
To the best of our knowledge, the problem of leader election with advice has never been studied before for anonymous networks.
In \cite{MP}, the authors investigated the minimum size of advice sufficient to find the largest-labelled node in a graph.
The main difference between  \cite{MP} and the present paper is that we consider networks without node labels. This is not a small difference:
from the methodological perspective, breaking symmetry in anonymous networks relies heavily on the structure of the graph, and, as far as results
are concerned, much more advice is needed.

\section{Terminology and preliminaries}

In this paper we use the word {\em path} to mean a simple path in the tree. For nodes $a$ and $b$, we denote by $d(a,b)$ the distance from $a$ to $b$, and by $path(a,b)$ the path $(a,\dots,b)$.
Nodes $a$ and $b$ are called the {\em endpoints} of this path. Let $b$ be a node in $path(a,c)$. We say that $path(a,c)$ is the {\em concatenation} of 
$path(a,b)$ and $path(b,c)$ and we write $path(a,c)=path(a,b)\cdot path(b,c)$.
The length of a path $P$, denoted by  $|P|$,  is the number of edges in it. 
Denote by  $seq(a,b)=(p_1,\dots , p_s)$ the sequence of all ports encountered when moving from $a$ to $b$ on $path(a,b)$. 
Odd-indexed terms in $seq(a,b)$ are called the {\em outgoing ports} of $seq(a,b)$.
We also use the operator $\cdot$ to denote the usual concatenation of sequences of integers, e.g., when concatenating two port sequences.

Let $v$ be a node of a tree $T$ and let $r$ be a non-negative integer. An {\em endless path in $V(v,r)$} is a simple path of length $r$, with endpoints $v$ and $v'$, such that $v'$ is not a leaf in $T$.  A {\em terminated path in $V(v,r)$} is a simple path of length at most $r$, with endpoints $v$ and $v'$, such that $v'$ is a leaf in $T$. See Figure \ref{paths} for examples of terminated and endless paths.

\begin{figure}[!ht]
\begin{center}
\includegraphics[scale=1]{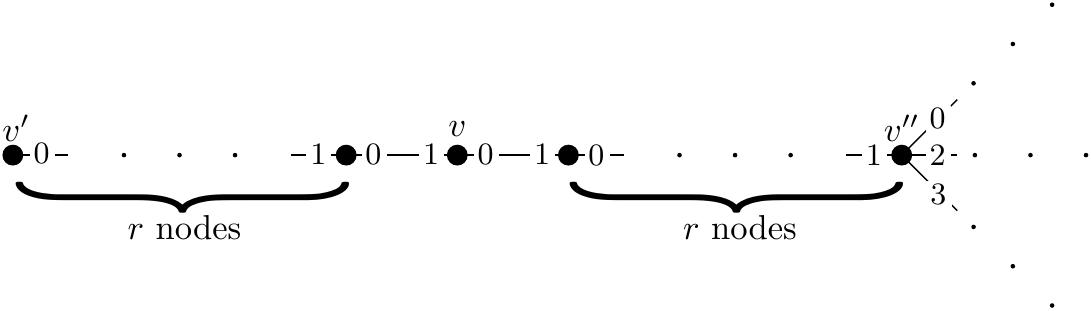}
\end{center}
\caption{An example of a node $v$'s view up to distance $r$, i.e. $V(v,r)$. The path between $v$ and $v'$ is a terminated path in $V(v,r)$, while the path between $v$ and $v''$ is an endless path in $V(v,r)$. All endless paths in $V(v,r)$ starting at $v$ pass through node $v''$.}
\label{paths}
\end{figure}

\section{Time $diam-1$}

In this section, we show tight upper and lower bounds of $\Theta (\log D)$ on the minimum size of advice sufficient to perform leader election in time $diam-1$ in
trees of diameter $diam \leq D$. The upper bound $O(\log D)$ is straightforward: given the value of $diam$, every node $v$ can reconstruct the entire tree from 
$V(v,diam-1)$ as follows. For each endless path with endpoints $v$, $v'$ , where $v'$ has some degree $k$, node $v$ attaches $k-1$ leaves to $v'$.
Hence, by using any centralized algorithm on the entire tree, all nodes can perform leader election whenever the tree is not symmetric. (This also shows that every non-symmetric tree $T$ has $\xi(T) \leq diam -1$.)
The matching lower bound is given by the following theorem.

\begin{theorem}\label{diam -1}
Consider any algorithm $ELECT$ such that, for every non-symmetric tree $T$, algorithm $ELECT$ solves election within $diam(T) - 1$ rounds. For every integer $D \geq 3$, there exists a tree $T$ with diameter at most $D$ and $\xi(T) \leq diam(T) -1$, for which algorithm $ELECT$ requires advice of size $\Omega(\log D)$.
\end{theorem}
\begin{proof}
Fix any integer $D \geq 3$. We will show a stronger statement: at least $D-1$ different advice strings are needed in order to solve election within $diam(T)-1$ rounds for some trees $T$ with diameter at most $D$ and $\xi (T) \leq diam(T)-2$. To prove this statement, we first construct a class of trees ${\cal T} = \{T_2,\ldots,T_D\}$, where tree $T_k$ is the path of length $k$. For each $k \in \{2,\ldots,D\}$, let $a_k$ and $b_k$ be the endpoints of $T_k$, and label the ports of $T_k$ such that the port sequence $seq(a_k,b_k)$ is equal to $(0,0,1,0,1,0,\ldots,1,0)$. See Figure \ref{logDpaths} for an illustration of $T_k$. We will denote by $P_{a_k}$ and $P_{b_k}$ the sequences of outgoing ports that are outputted by $a_k$ and $b_k$, respectively, at the end of the execution of algorithm $ELECT$ in tree $T_k$. Note that algorithm $ELECT$ is correct  only if, for every $k \in \{2,\ldots,D\}$, there exists a node $\ell_k \in T_k$ such that sequence $P_{a_k}$ corresponds to a simple path from $a_k$ to $\ell_k$, and sequence $P_{b_k}$ corresponds to a simple path from $b_k$ to $\ell_k$.  Hence, algorithm $ELECT$ is correct  only if $|P_{a_k}| + |P_{b_k}| = k$ for each $k \in \{2,\ldots,D\}$.

\begin{figure}[!ht]
\begin{center}
\includegraphics[scale=1]{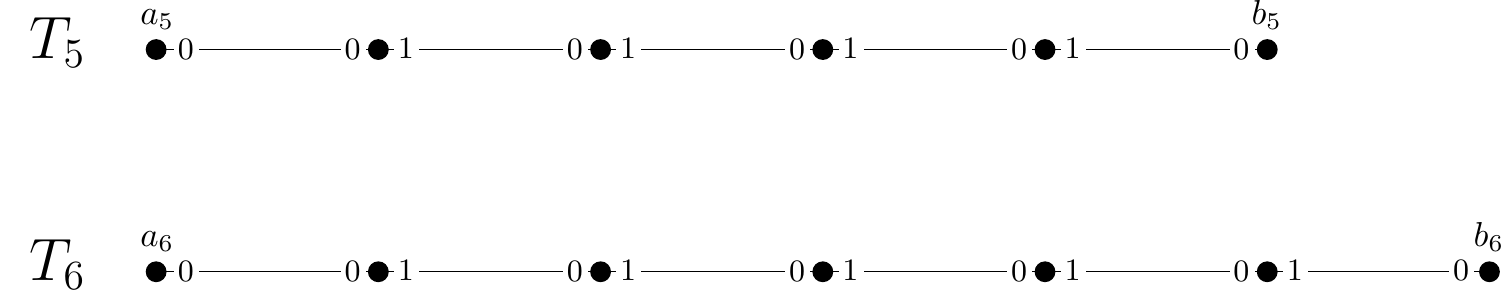}
\end{center}
\caption{Example of trees $T_k$ constructed in Theorem \ref{diam -1}, for $k=5,6$.}
\label{logDpaths}
\end{figure}

Next, to obtain a contradiction, assume that $D-2$ different advice strings are sufficient to solve election within $diam(T)-1$ rounds for each tree $T$ in ${\cal T}$. By the Pigeonhole Principle, there exist $i,j \in \{2,\ldots,D\}$ with $i<j$ such that the same advice string is provided to the nodes of both $T_i$ and $T_j$ when they execute algorithm $ELECT$. When executed at node $a_i$ in $T_i$, algorithm $ELECT$ halts in some round $r_a \leq diam(T_i)-1 = i-1$ and outputs some port sequence $P_{a_i}$. Similarly, when executed at node $b_i$ in $T_i$, algorithm $ELECT$ halts in some round $r_b \leq i-1$ and outputs some port sequence $P_{b_i}$. As noted above, we have that $|P_{a_i}| + |P_{b_i}| = i$. We show that, when executed at node $a_j$ in $T_j$, algorithm $ELECT$ also halts in round $r_a$ and outputs $P_{a_i}$. Indeed, the algorithm is provided with the same advice string for both $T_i$ and $T_j$, and $V_{T_i}(a_i,r_a)=V_{T_j}(a_j,r_a)$. Similarly, when executed at node $b_j$ in $T_j$, algorithm $ELECT$ halts in round $r_b$ and outputs $P_{b_i}$. However, this implies that, in the execution of $ELECT$ in tree $T_j$, we have $|P_{a_j}| + |P_{b_j}| = |P_{a_i}| + |P_{b_i}| =  i < j$, which contradicts the correctness of $ELECT$.

We finally show that, for every $k \in \{2,\ldots,D\}$, we have $\xi(T_k) \leq k-2$. First assume that $k>2$.
For every node $v$ of $T_k$, at least one of the endpoints of $T_k$ is in $V(v, k-2)$. Hence $v$ can identify its position in the map of $T_k$ and
output the sequence of outgoing ports leading from $v$ to $a_k$. For $k=2$, both leaves output the sequence $(0)$ and the central node outputs the empty sequence.
\end{proof}

\section{Time $diam-2$}

In this section, we show tight upper and lower bounds on the minimum size of advice sufficient to perform leader election in time $diam-2$, for
trees of diameter $diam \leq D$.
These bounds depend on the parity of $diam$. They are $\Theta(\log D)$ for even values of $diam$, and $\Theta(\log n)$ for odd values of $diam$. We consider these two cases separately. 

\subsection{Even Diameter}
Consider any tree $T$ with $n$ nodes and even diameter $diam \leq D$. The lower bound $\Omega (\log D)$ on the minimum size of advice sufficient to perform leader election in time $diam-2$ can be proven exactly as  Theorem \ref{diam -1}. We now prove the matching upper bound by  
providing an algorithm {\tt EvenElect} that solves election in time $diam-2$ using $O(\log{D})$ bits of advice. 
The algorithm works by having each node find and elect the central node of the tree, which we denote by $v_c$.
The advice provided to the algorithm is the value of $diam$.  Let $h = diam/2$.  The {\em gateway} $g_v$ of a node $v$ is defined as the node in  $V(v,diam-2)$ furthest from $v$ such that every endless path in $V(v,diam-2)$ passes through $g_v$. We now give the pseudo-code of the algorithm executed at an arbitrary node $v$ in tree $T$ using advice $A$.

\begin{algorithm}[H]
\caption{\texttt{EvenElect}($A$)}
\begin{algorithmic}[1]
\State $v_c \leftarrow \emptyset$
\State $diam \leftarrow$ diameter of $T$, as provided in $A$
\State $h \leftarrow diam/2$
\State Use $diam-2$ rounds of communication to learn $V(v,diam-2)$
\State {\bf If} $V(v,diam-2)$ contains no endless paths starting at $v$:
\State \indent $v_c \leftarrow$ central node of $V(v,diam-2)$
\State {\bf Else}:
\State \indent  $g_v \leftarrow$ the node $w$ in $V(v,diam-2)$ furthest from $v$ such that every endless path in $V(v,diam-2)$ passes through $w$ \label{startcalc}
\State \indent {\bf If} $d(v,g_v) \leq h-1$, or, $V(v,diam-2)$ contains a node $w$ such that $d(w,g_v) > d(v,g_v)$:
\State \indent \indent $\ell \leftarrow h-1$
\State \indent {\bf Else}:
\State \indent \indent $\ell \leftarrow h$ \label{endcalc}
\State \indent $v_c \leftarrow$ the node on $path(v,g_v)$ at distance $\ell$ from $v$ \label{vc}
\State Output the sequence of outgoing ports of $seq(v,v_c)$
\end{algorithmic}
\label{nodeelect}
\end{algorithm}

\begin{theorem}\label{even}
Algorithm {\tt EvenElect} solves leader election in trees of size $n$ and even diameter $diam \leq D$ in time $diam-2$ using $O(\log{D})$ bits of advice.
\end{theorem}
\begin{proof}
We begin by proving the correctness of the algorithm. In particular, we must show that each node $v$ correctly computes the central node $v_c$. If $V(v,diam-2)$ contains no endless paths starting at $v$, then $V(v,diam-2)$ consists of the entire tree. In this case, $v$ can find the central node by inspection. Otherwise, it follows that $d(v,v_c) \in \{h-1,h\}$. The following result shows that $v_c$ always lies on the path from $v$ to $g_v$.

\begin{claim}\label{nodeonpath}
If $v$ is a node such that $d(v,v_c) \in \{h-1,h\}$, then $v_c$ is on $path(v,g_v)$.
\end{claim}
To prove the claim, 
let $Q_1 = path(v,v_c)$, and let $Q_2$ be any path starting from $v_c$ of length $h$ such that $Q_2$ and $path(v_c,v)$
are edge-disjoint. 
Let $Q$ be the path $Q_1 \cdot Q_2$. Note that $|Q| \geq (h-1) + h = 2h-1 = diam-1$. In particular, in $V(v,diam-2)$, $Q$ is an endless path starting at $v$. Therefore, by the definition of $g_v$, node $g_v$ is on path $Q$. If $g_v$ is in $Q_2$, then $v_c$ is on $path(v,g_v)$, and we are done. Finally, we show that $g_v$ cannot appear before $v_c$ in $Q$, by way of contradiction. Assume it does. By the definition of $g_v$, there must exist some endless path in $V(v,diam-2)$ that contains $g_v$ but not $v_c$. In particular, there must be a path $Q'$ of length at least $diam-1$ with $v$ as one endpoint that passes through $g_v$ but not through $v_c$. Let $v'$ be the endpoint of $Q'$ not equal to $v$, and let $w$ be the node in $Q'$ that is furthest from $v$ and on $path(v,v_c)$. We consider the length of the path $Q'' = path(v',w) \cdot path(w,v_c) \cdot Q_2$. (See Figure \ref{proofpathseven} for an illustration of the paths defined above.) Since $|Q'| \geq diam-1 = 2h-1$, it follows that $|path(v,w)| + |path(w,v')| \geq 2h-1$. Also, note that $|path(v,v_c)| = |path(v,w)| + |path(w,v_c)|$. So, $|Q''| = |path(v',w)| + |path(w,v_c)| + |Q_2| \geq (2h - 1 - |path(v,w)|) + (|path(v,v_c)| - |path(v,w)|) + h = (3h - 1) - d(v,w) + (d(v,v_c) - d(v,w))$. Finally, since $w$ appears before $v_c$ in $Q$, we have $d(v,w) < d(v,v_c) \leq h$. Therefore, $|Q''| \geq (3h - 1) - (h-1) + 1= 2h + 1 > diam$, a contradiction. This concludes the proof of Claim \ref{nodeonpath}.

\begin{figure}[!ht]
\begin{center}
\includegraphics[scale=0.8]{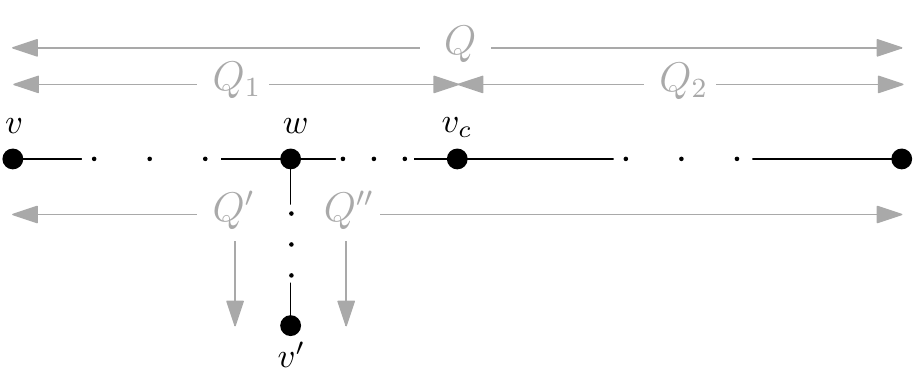}
\end{center}
\caption{Paths $Q_1,Q_2,Q,Q',Q''$ as defined in the proof of Claim \ref{nodeonpath}.}
\label{proofpathseven}
\end{figure}


We now show that $v$ correctly computes $d(v,v_c)$ (which it stores in $\ell$) at lines \ref{startcalc}-\ref{endcalc}.
\begin{claim}\label{case1}
Suppose that $d(v,v_c) \in \{h-1,h\}$. If $d(v,g_v) \leq h-1$, or, $V(v,diam-2)$ contains a node $w$ such that $d(w,g_v) > d(v,g_v)$, then $d(v,v_c) = h-1$.
\end{claim}
To prove the claim, first suppose that $d(v,g_v) \leq h-1$. It follows from Claim \ref{nodeonpath} that $d(v,v_c) \leq d(v,g_v) \leq h-1$, which implies that $d(v,v_c) = h-1$. Next, suppose that $V(v,diam-2)$ contains a node $w$ such that $d(w,g_v) > d(v,g_v)$. It follows from Claim \ref{nodeonpath} that $d(v,g_v) = d(v,v_c) + d(v_c,g_v)$, so $d(v,v_c) + d(v_c,g_v) < d(w,g_v) \leq d(w,v_c) + d(v_c,g_v)$. Hence, $d(v,v_c) < d(w,v_c)$. Note that, by the definition of the central node, we have $d(w,v_c) \leq h$. So, we have shown that $d(v,v_c) \leq h-1$, which implies that $d(v,v_c) = h-1$. This completes the proof of the claim.

\begin{claim}\label{case2}
Suppose that $d(v,v_c) \in \{h-1,h\}$. lf $d(v,g_v)>h-1$, and, for all $w \in V(v,D-2)$, $d(w,g_v) \leq d(v,g_v)$, then $d(v,v_c) = h$.
\end{claim}
We prove the contrapositive of this claim. Namely, we show that, if $d(v,v_c) = h-1$ and, for all $w \in V(v,diam-2)$, $d(w,g_v) \leq d(v,g_v)$, then $d(v,g_v)\leq h-1$. Let $e_v$ be the first edge on the path from $v_c$ to $v$. Let $T_v$ be the subtree induced by all nodes reachable from $v_c$ via a path starting with edge $e_v$.  For each node $w \in T_v$, we have $path(w,g_v) = path(w,v_c)\cdot path(v_c,g_v)$ since $v_c$ lies on the path from $v$ to $g_v$. Hence, $d(w,v_c) = d(w,g_v) - d(v_c,g_v) \leq d(v,g_v) - d(v_c,g_v) = d(v,v_c) = h-1$, i.e., every path from $v_c$ to a node in $T_v$ has length at most $h-1$. By the definition of the central node, there must exist two distinct edges incident to $v_c$ that belong to paths $P_1, P_2$ of length $h$ starting at $v_c$. Neither of these edges is equal to $e_v$, since $T_v$ does not contain a path
of length $h$ starting at $v_c$.
Let $z_1,z_2$ be the leaves of  paths $P_1, P_2$, respectively. Since $d(v,z_1) = d(v,z_2)= d(v,v_c)+h = 2h-1$, it follows that $path(v,z_1)$ and $path(v,z_2)$ are endless paths in $V(v,2h-2) = V(v,diam-2)$. Rewrite $path(v,z_1) = path(v,v_c)\cdot path(v_c,z_1)$ and $path(v,z_2) = path(v,v_c) \cdot path(v_c,z_2)$, and recall that the first edge of $path(v_c,z_1)$ is not equal to the first edge of $path(v_c,z_2)$. By the definition of $g_v$, it follows that $g_v \in path(v,v_c)$. Therefore, $d(v,g_v) \leq h-1$, which proves the claim.

By Claim \ref{nodeonpath}, $v_c$ is precisely the node at distance $d(v,v_c)$ from $v$ on $path(v,g_v)$.  By Claims \ref{case1} and \ref{case2}, we see that the value of $\ell$ computed in the algorithm is equal to $d(v,v_c)$. It follows that $v$ correctly calculates $v_c$ at line \ref{vc}.

Finally, note that the advice consists of the value of $diam$, so the size of advice is $O(\log{D})$.
\end{proof}

{\bf Note.} Theorem \ref{even} implies that, for any tree $T$ of even diameter $diam$, we have $\xi(T) \leq diam -2$.

\subsection{Odd Diameter}

We now provide tight upper and lower bounds of $\Theta( \log n)$ on the minimum size of advice sufficient to perform leader election in time $diam-2$, when $diam$ is odd.
Our lower bound is valid even for the class of trees with fixed diameter $D$. For the upper bound, we can provide the value of the diameter
as part of the advice. 

Consider any tree $T$ with $n$ nodes and odd diameter $D$. We prove our upper bound by providing an algorithm {\tt OddElect} that solves election in time $D-2$ using $O(\log{n})$ bits of advice.

At a high level, our algorithm works as follows. Using $D-2$ communication rounds, each node calculates its simple path to the closest endpoint of the central edge. For each node $v$, this closest endpoint  will be called $v$'s \emph{candidate}. 
The main difficulty of the algorithm is breaking symmetry between the two candidates.
The advice helps the nodes decide which of the two possible candidates should be elected as leader, and provides the port number which leads from the non-elected candidate to the leader. To do this with a small number of bits, the advice succinctly describes a path which exists starting at one of the two endpoints of the central edge but not the other. The nodes that see this path starting from their candidate will elect their candidate, and the nodes that cannot see this path starting from their candidate will use the port number provided in the advice to elect the other candidate.

We now provide the details of the advice and the algorithm. Let $\{c_0,c_1\}$ be the central edge of the tree. For each node $v$, denote by $cand(v)$ the node in $\{c_0,c_1\}$ that is closest to $v$. Let $h = \frac{D-1}{2}$. Note that, by the definition of the central edge, for each node $v$, we have $d(v,cand(v)) \leq h$. Recall that the gateway $g_v$ of a node $v$ is defined as the node in  $V(v,D-2)$ furthest from $v$ such that every endless path in $V(v,D-2)$ passes through $g_v$.

We first construct the advice.
The first part of the advice string is the exact value of $D$, which can be used to calculate the value of $h$. The goal of the rest of the advice construction is to succinctly describe a sequence of port numbers that distinguishes one of the two candidate nodes from the other. 

We divide the set of trees into two classes. We say that a tree is \emph{separated} if, for each node $v$,  $path(v,g_v)$ contains the central edge. One bit of the advice string, called the ``separated bit'', has value 1 if and only if $T$ is separated.

If $T$ is separated, we construct the remainder of the advice string as follows. For each of the two endpoints $c_0,c_1$ of the central edge, we define the list $L_i$ consisting of all port sequences that can be obtained by following simple paths starting at $c_i$ that do not contain the central edge. These port sequences consist of both the outgoing and incoming port numbers encountered, in order, on each path. Each list $L_i$ is sorted in ascending lexicographic order. The following result shows that $L_0$ and $L_1$ must differ.

\begin{claim}\label{diff}
For any separated tree $T$, if $L_0 = L_1$, then $T$ is symmetric.
\end{claim}
To prove the claim, it suffices to note that, if $L_0=L_1$, then the subtrees rooted at $c_0$ and $c_1$ and resulting from the removal of the central edge are
(port-preserving) isomorphic. It follows that the port numbers at the two endpoints of the central edge must be equal. Hence $T$ is symmetric,
which proves the claim.

By Claim \ref{diff}, there exists a port sequence that appears in exactly one of  $L_0$ or $L_1$. Formally,  for some $i \in \{0,1\}$, there is an integer $j$ such that the $j^{th}$ sequence in $L_i$ does not appear in $L_{1-i}$. The remainder of the advice string is a tuple $(j,k,m,p)$ where:
\begin{itemize}
\item $k$ is the largest integer such that the $j^{th}$ sequences in $L_i$ and $L_{1-i}$ have equal prefixes of length~$k$,
\item $m$ is the integer equal to the $(k+1)^{th}$ port number of the $j^{th}$ sequence in $L_i$, and,
\item $p$ is the port number that leads from $c_{1-i}$ to $c_i$.
\end{itemize}

We now describe the advice string in the case where $T$ is not separated. The construction is similar to the case where $T$ is separated, except for a change in the definition of the lists $L_0$ and $L_1$. In particular, for each of the two endpoints $c_0,c_1$ of the central edge, we define list $L_i$ to be all port sequences of length at most $2h-1$ that can be obtained by following simple paths starting at $c_i$. As before, these port sequences consist of both the outgoing and incoming port numbers encountered, in order, on each path. Again, each list $L_i$ is sorted in ascending lexicographic order. The following result shows that, also in the case of non-separated trees, lists $L_0$ and $L_1$ must differ.

\begin{claim}\label{diff2}
For any tree $T$ that is not separated, if $L_0 = L_1$, then $\xi(T) > D-2$.
\end{claim}
Our proof of the claim proceeds in three steps. First, we find two leaves $w_0,w_1$ such that $w_0$'s candidate node is $c_0$, $w_1$'s candidate node is $c_1$, and $seq(w_0,c_0) = seq(w_1,c_1)$. We then show that $V(w_0,D-2) = V(w_1,D-2)$. Finally, we show that this implies that $\xi(T) > D-2$. In what follows, for any sequence $s$, we will denote by $\bar{s}$ the reverse of sequence $s$. 

{\bf Finding $w_0$ and $w_1$.} We first note that, since $T$ is not separated, there must be at least two nodes, say $\alpha,\beta$, such that $path(\alpha,\beta)$ does not use the central edge and has length $2h$. It follows that $\alpha$ and $\beta$ have the same candidate node, which, without loss of generality, we assume is $c_0$. Further, it follows that $d(\alpha,c_0) = d(\beta,c_0)=h$, and that the last port numbers in $seq(\alpha,c_0)$ and $seq(\beta,c_0)$ are different. Let $p_1$ be the port number at $c_1$ corresponding to the central edge. Let $w_0$ be a node in $\{\alpha,\beta\}$ such that the last port number in $seq(w_0,c_0)$ is not equal to $p_1$. We now set out to find a node $w_1$ such that $cand(w_1) = c_1$ and $seq(w_1,c_1) = seq(w_0,c_0)$. Since $d(c_0,w_0) = h$, it follows that $w_0$ is a leaf, so $seq(c_0,w_0)$ is a sequence of length $2h$ with last port number equal to 0. Let $\sigma$ be the prefix of length $2h-1$ of $seq(c_0,w_0)$. Since $\sigma$ is a sequence of length $2h-1$ that can be obtained by following a simple path starting at $c_0$, we know that $\sigma$ appears in $L_0$. Since $L_0=L_1$, it follows that $\sigma$ also appears in $L_1$. Let $w_1$ be the node that is reached by following the outgoing ports of $\sigma$ starting at $c_1$. By our choice of $w_0$, the first port number in $\sigma$ is not equal to $p_1$, so $path(c_1,w_1)$ does not use the central edge. It follows that $cand(w_1) = c_1$. Also, since $w_0$ and $w_1$ are leaves, it follows that $seq(w_0,c_0) = (0)\cdot \bar{\sigma} = seq(w_1,c_1)$. This completes the first step of the proof.

In what follows, let $e_0$ be the first edge on the path from $c_0$ to $w_0$, and let $e_1$ be the first edge on the path from $c_1$ to $w_1$. Let $T_0$ be the subtree induced by all nodes that can be reached by a simple path starting with edge $e_0$, and let $T_1$ be the subtree induced by all nodes that can be reached by a simple path starting with edge $e_1$. Since all nodes in $T_0$ (respectively, $T_1$) are at distance at most $h$ from $c_0$ (respectively, $c_1$), the fact that $L_0 = L_1$ implies that there is a port-preserving isomorphism between $T_0$ and $T_1$. Since $seq(w_0,c_0) = seq(w_1,c_1)$, it follows that such an isomorphism maps $w_0$ to $w_1$.

{\bf Showing that $V(w_0,D-2) \subseteq V(w_1,D-2)$} (a symmetric argument proves the reverse inclusion.) Since $D-2 = 2h-1$, it suffices to show that each sequence $\phi$ of at most $4h-1$ port numbers obtained by following a simple path $P_{\phi}$ starting at $w_0$ can also be obtained by following a simple path starting at $w_1$. Let $\phi$ be any such sequence. We consider two cases. First, suppose that $P_{\phi}$ does not contain $c_0$. Then $P_{\phi}$ lies entirely within $T_0$. Since there is a port-preserving isomorphism between $T_0$ and $T_1$ that maps $w_0$ to $w_1$, the same path exists in $T_1$ starting at $w_1$, as desired. Next, suppose that $P_{\phi}$ does contain $c_0$. We re-write $\phi = seq(w_0,c_0) \cdot \phi'$ for some port sequence $\phi'$.  
Since $seq(w_0,c_0) = seq(w_1,c_1)$ and $|seq(w_0,c_0)| = 2h$, it follows that the first $2h$ ports of $\phi$ form $seq(w_1,c_1)$. 
Further, since $|\phi| \leq 4h-1$, it follows that $|\phi'| \leq 2h-1$. Since $\phi'$ corresponds to a path starting at $c_0$, the sequence $\phi'$ appears in $L_0$. Since $L_0 = L_1$, we know that $\phi'$ appears in $L_1$, so $\phi'$ is a sequence of port numbers that can be obtained by following a simple path starting at $c_1$. Hence, $\phi$ can be obtained by following a simple path starting at $w_1$, as desired. This completes the second step of the proof. 

{\bf Showing that $V(w_0,D-2) = V(w_1,D-2)$ implies that $\xi(T) > D-2$.} To obtain a contradiction, suppose that $V(w_0,D-2)= V(w_1,D-2)$ and assume that there is an algorithm that solves election in $T$ within $D-2$ rounds (with any amount of advice). For any such algorithm, nodes $w_0$ and $w_1$ output the same value since $V(w_0,D-2) = V(w_1,D-2)$. In particular, they both output outgoing port sequences of equal length, say $\ell$. Since $w_0$ and $w_1$ have different candidate nodes, and both $w_0$ and $w_1$ must elect the same node, it follows that at least one of the paths obtained by following their outputs must cross the central edge. Hence, $\ell > h$. However, each of their outputs forms a path of length $\ell$ ending at the elected node. These two paths combine to form a simple path of length $2\ell \geq 2(h+1) > D$, a contradiction. This concludes the proof of the claim.

By Claim \ref{diff2}, lists $L_0$ and $L_1$ must differ in the case that $\xi(T) \leq D-2$, i.e., when leader election is possible in time $D-2$. The remainder of the advice string consists of the tuple $(j,k,m,p)$ as defined in the case of separated trees. This concludes the description of the advice. 

We now give the pseudo-code of the algorithm executed at an arbitrary node $v$ in tree $T$ using advice $A$.

\begin{algorithm}[H]
\caption{\texttt{OddElect}($A$)}
\begin{algorithmic}[1]
\State $e \leftarrow \emptyset$
\State $D \leftarrow$ diameter of $T$, as provided in $A$
\State $h \leftarrow (D-1)/2$
\State Use $D-2$ rounds of communication to learn $V(v,D-2)$
\State // Stage 1: compute $cand(v)$
\State {\bf If} $V(v,D-2)$ contains no endless paths starting at $v$:
\State \indent $e \leftarrow$ central edge of $V(v,D-2)$
\State \indent $cand(v) \leftarrow$ endpoint of $e$ closest to $v$
\State {\bf Else}:
\State \indent  $g_v \leftarrow$ the node $w$ in $V(v,D-2)$ furthest from $v$ such that every endless path in $V(v,D-2)$ passes through $w$
\State \indent {\bf If} $V(v,D-2)$ contains a node $w$ such that $d(w,g_v) > d(v,g_v)$:
\State \indent \indent $\ell \leftarrow h-1$
\State \indent {\bf Else}:
\State \indent \indent $\ell \leftarrow h$
\State \indent $cand(v) \leftarrow$ the node on $path(v,g_v)$ at distance $\ell$ from $v$
\State // Stage 2: determine whether or not $cand(v)$ should be elected as leader
\State {\bf If} the ``separated bit'' in $A$ is 1:
\State \indent {\bf If} $e = \emptyset$:
\State \indent \indent $e \leftarrow$ the edge incident to $cand(v)$ that lies on all endless paths starting at $v$ \label{findcentral}
\State \indent Compute $seq(v,v')$ for each $v'$ such that $path(v,v')$ does not contain edge $e$ \label{avoidcentral}
\State \indent $L_v \leftarrow$ the lexicographically-ordered list of all such $seq(v,v')$\label{defineL_v}
\State {\bf Else}:
\State \indent Compute every port sequence of length at most $2h-1$ corresponding to paths starting at $cand(v)$ \label{computepaths}
\State \indent $L_v \leftarrow$ the lexicographically-ordered list of all such sequences
\State Retrieve $j,k,m$ from $A$
\State {\bf If} the $(k+1)^{th}$ port number of the $j^{th}$ sequence in $L_v$ is equal to $m$:
\State \indent $electMyCandidate \leftarrow true$
\State {\bf Else}:
\State \indent $electMyCandidate \leftarrow false$
\State // Stage 3: compute output
\State {\bf If} $electMyCandidate = true$:
\State \indent Output the sequence of outgoing ports of $seq(v,cand(v))$
\State {\bf Else}:
\State \indent Retrieve $p$ from $A$
\State \indent Output the sequence of outgoing ports of $seq(v,cand(v))$ with $p$ appended to the end
\end{algorithmic}
\label{edgeelect}
\end{algorithm}

\begin{theorem}\label{ub-log}
Algorithm {\tt OddElect} solves leader election in trees $T$ with size $n$, odd diameter $D$ and $\xi(T)\leq D-2$, in time $D-2$ using $O(\log{n})$ bits of advice.
\end{theorem}
\begin{proof}
We begin by proving the correctness of each stage of the algorithm.

{\bf First stage.} We must show that each node $v$ correctly computes $cand(v)$. If $V(v,D-2)$ contains no endless paths starting at $v$, then $V(v,D-2)$ consists of the entire tree. In this case, $v$ can find the central edge by inspection and calculate $cand(v)$ as the closest endpoint of the central edge. Otherwise, it follows that $d(v,cand(v)) \in \{h-1,h\}$. The following result shows that $cand(v)$ always lies on the path from $v$ to $g_v$.

\begin{claim}\label{onpath}
If $v$ is a node such that $d(v,cand(v)) \in \{h-1,h\}$, then $cand(v)$ is on $path(v,g_v)$.
\end{claim}
To prove the claim, assume, without loss of generality, that $cand(v) = c_0$. Let $Q_1 = path(v,c_0)$, and let $Q_2$ be any path starting from $c_1$ of length $h$ that does not use the central edge. Let $Q$ be the path $Q_1\cdot (c_0,c_1) \cdot Q_2$. Note that $|Q| \geq (h-1) + 1 + h = 2h = D-1$. In particular, in $V(v,D-2)$, $Q$ is an endless path. Therefore, by the definition of $g_v$, node $g_v$ is on path $Q$. If $g_v $ is in $(c_0,c_1)\cdot Q_2$, then $c_0$ is on $path(v,g_v)$, and we are done. Finally, we show that $g_v$ cannot appear before $c_0$ in $Q$, by way of contradiction. Assume it does. By the definition of $g_v$, there must exist some endless path in $V(v,D-2)$ that contains $g_v$ but not $c_0$. In particular, there must be a path $Q'$ of length at least $D-1$ with $v$ as one endpoint that passes through $g_v$ but not $c_0$. Let $v'$ be the endpoint of $Q'$ not equal to $v$, and let $w$ be the node in $Q'$ that is furthest from $v$ and on $path(v,c_0)$. We consider the length of the path $Q'' = path(v',w) \cdot path(w,c_0) \cdot (c_0,c_1) \cdot Q_2$. (See Figure \ref{proofpathsodd} for an illustration of the paths defined above.) Since $|Q'| \geq D-1 = 2h$, it follows that $|path(v,w)| + |path(w,v')| \geq 2h$. Also, note that $|path(v,c_0)| = |path(v,w)| + |path(w,c_0)|$. So, $|Q''| \geq (2h - |path(v,w)|) + (|path(v,c_0)| - |path(v,w)|) + 1 + h = 3h + 1 + d(v,cand(v)) - 2\cdot d(v,w)$. Finally, since $w$ appears before $c_0$
in $Q$, we have $d(v,w) < d(v,c_0) \leq h$. Therefore, $|Q''| \geq 3h + 1 + (h-1) - 2(h-1) = 2h + 2 > D$, a contradiction. This concludes the proof of Claim \ref{onpath}.

\begin{figure}[!ht]
\begin{center}
\includegraphics[scale=0.8]{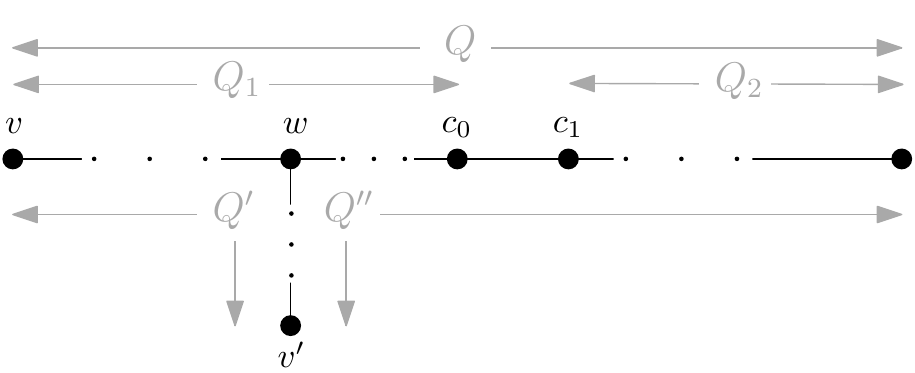}
\end{center}
\caption{Paths $Q_1,Q_2,Q,Q',Q''$ as defined in the proof of Claim \ref{onpath}.}
\label{proofpathsodd}
\end{figure}

We will use the following claim to show that each $v$ can determine its exact distance from $cand(v)$.

\begin{claim}\label{finddist}
Consider any node $v$ such that $d(v,cand(v)) \in \{h-1,h\}$. There exists a node $w$ in $V(v,D-2)$ such that $d(w,g_v) > d(v,g_v)$ if and only if $d(v,cand(v)) = h-1$.
\end{claim}

We first prove the `` if'' direction. Suppose that $d(v,cand(v)) = h-1$. Let $w$ be a node such that $cand(w) = cand(v)$ and $d(w,cand(w)) = h$. The existence of $w$ is guaranteed by the definition of the central edge. Thus, we have $d(v,cand(v)) < d(w,cand(v))$. Further, we have $d(v,w) \leq d(v,cand(v)) + d(cand(v),w) = 2h-1 = D-2$, so $w$ is a node in $V(v,D-2)$. By Claim \ref{onpath}, node $cand(v)$ is on $path(v,g_v)$, so $d(v,g_v) = d(v,cand(v)) + d(cand(v),g_v) < d(w,cand(v)) + d(cand(v),g_v) = d(w,cand(w)) + d(cand(w),g_v) \leq d(w,g_v)$, as required. 

Next we prove the ``only if'' direction. Suppose that $d(v,cand(v)) = h$. Let $w$ be any node in $V(v,D-2)$. In particular, this means that $d(v,w) \leq D-2 = 2h-1$. There are two cases to consider based on whether or not $cand(v)=cand(w)$, i.e., whether or not $v$ and $w$ are on the same side of the central edge. In the first case, suppose that $cand(v)=cand(w)$. It follows that $d(w,cand(v)) \leq h = d(v,cand(v))$, so $d(w,g_v) \leq d(w,cand(v)) + d(cand(v),g_v) \leq d(v,cand(v)) + d(cand(v),g_v) = d(v,g_v)$ (where the last equality follows from Claim \ref{onpath}). Hence, $d(v,g_v) \geq d(w,g_v)$, as required. In the second case, suppose that $cand(v) \neq cand(w)$. It follows that $path(v,cand(v))$  and $path(w,cand(v))$ intersect only at node $cand(v)$ since the shortest path from $w$ to $cand(v)$ passes through $cand(w)$ first. Therefore, $d(v,w) = d(v,cand(v)) + d(cand(v),w)$. Since $w$ is in $V(v,D-2)$, it follows that $d(v,w) \leq D-2 = 2h-1$, so $d(cand(v),w) = d(v,w) - d(v,cand(v)) \leq (2h-1) - h = h-1$. Thus, $d(w,g_v) \leq d(w,cand(v)) + d(cand(v),g_v) \leq (h-1) + d(cand(v),g_v) < d(v,cand(v))+ d(cand(v),g_v) = d(v,g_v)$ (where the last equality follows from Claim \ref{onpath}). Hence, $d(v,g_v) \geq d(w,g_v)$, as required. This concludes the proof of the claim.

By Claim \ref{onpath}, $cand(v)$ is precisely the node at distance $d(v,cand(v))$ from $v$ on $path(v,g_v)$.  By Claim \ref{finddist}, we see that the value of $\ell$ computed in the first stage of the algorithm is equal to $d(v,cand(v))$. It follows that $v$ correctly calculates $cand(v)$ during stage 1.

{\bf Second stage.}
To prove the correctness of this stage of the algorithm, we show that if $cand(v) = c_i$ for some $i \in \{0,1\}$, then $L_v = L_i$ (as defined in the advice construction). In the case where $T$ is a separated tree, we see that the construction on lines \ref{avoidcentral} and \ref{defineL_v} matches the definition of $L_a$, as long as
$e$ is the central edge.
If the central edge was assigned to $e$ during stage 1, then $e$ is still the central edge at line \ref{avoidcentral}. Otherwise, we must show that the central edge is assigned to $e$ at line \ref{findcentral}. The following result confirms that this is the case.

\begin{claim}
Suppose that $T$ is separated and consider any node $v$ such that $V(v,D-2)$ contains at least one endless path starting at $v$. The central edge of $T$ is the edge incident to $cand(v)$ that lies on all endless paths starting at $v$.
\end{claim}
To prove the claim, recall that, since $T$ is separated, $path(v,g_v)$ contains the central edge. Further, by the definition of $g_v$, every endless path starting at $v$ passes through $g_v$. It follows that every endless path starting at $v$ contains the central edge, which proves the claim.

In the case where $T$ is not a separated tree, it suffices to show that line \ref{computepaths} is possible to carry out, i.e., that each node $v$ can compute every port sequence of length at most $2h-1$ corresponding to paths starting at $cand(v)$. To see why this is the case, note that $D-2 = 2h-1$, so each node $v$ knows $V(v,2h-1)$. Since $d(v,cand(v)) \leq h$ and $v$ has computed $cand(v)$, it follows that $v$ also knows $V(cand(v),h-1)$. From $V(cand(v),h-1)$, $v$ can compute all port sequences of length at most $2h-2$ corresponding to paths of length $h-1$ starting at $cand(v)$, and, using the degrees of nodes at distance $h-1$ from $cand(v)$, $v$ can compute all port sequences of length at most $2h-1$ corresponding to paths starting at $cand(v)$.
This matches the definition of $L_i$ in the case when the tree is not separated. Therefore, regardless of whether $T$ is separated or not, the list $L_v$ is identical to the list $L_i$ corresponding to $v$'s candidate $c_i$, as produced in the advice construction.

{\bf Third stage.}
To prove the correctness of this stage of the algorithm, we show that all nodes output a sequence of outgoing ports leading to the same node. Without loss of generality, assume that, in the advice construction, the $(k+1)^{th}$ port of the $j^{th}$ sequence in $L_0$ is equal to $m$, while the $(k+1)^{th}$ port of the $j^{th}$ sequence in $L_1$ is not equal to $m$. Hence, $p$ is defined as the port leading from $c_1$ to $c_0$. We will show that the output of every node leads to node $c_0$. We showed above that, in stage 2 of the algorithm, each node $v$ with $cand(v) = c_0$ sets $L_v$ to $L_0$, so each such $v$ sets $electMyCandidate = true$. Moreover, each node $v$ with $cand(v) = c_1$ sets $L_v$ to $L_1$, so each such $v$ sets $electMyCandidate = false$. It follows that each node $v$ with $cand(v) = c_0$ outputs a sequence of outgoing ports leading to $c_0$, while each node $v$ with $cand(v) = c_1$ outputs a sequence of outgoing ports leading to $c_1$, with port $p$ appended. However, $p$ leads from $c_1$ to $c_0$, so the output of each node $v$ with $cand(v) = c_1$ leads to $c_0$. This concludes the proof of correctness.


Finally, note that the advice consists of 1 ``separated bit'', the value of $D$, and four integers whose values are each bounded above by $n$. Thus, the size of advice is $O(\log{n})$.
\end{proof}

We finish this section by proving the matching lower bound of $\Omega (\log n)$ on the minimum size of advice sufficient to perform leader election in time $diam-2$, for odd values of $diam$. As a first step, we analyze the following abstract {\em pair breaking} problem with parameter $Z$, where $Z$ is a positive integer. Denote by $X$ the set of pairs $(a,b)$,
where $a,b \in \{1,\dots , Z\}$ and $a < b$. The set $X$ is coloured with $c$ colours by a colouring function ${\cal C} : X \longrightarrow \{1, \dots ,c\}$. 
Knowing $Z$ and ${\cal C}$, the goal is to find a function ${\cal B}: \{1, \dots , Z\} \times \{1, \dots ,c\} \longrightarrow \{0,1\}$ with the following property:
${\cal B}(a,\gamma)\neq {\cal B}(b,\gamma)$, where $\gamma = {\cal C}((a, b))$. Such functions will be called {\em  symmetry breaking} functions.  What is the minimum integer $c$ for which there exists a
 colouring function ${\cal C} : X \longrightarrow \{1, \dots ,c\}$  such that this goal is attainable?
 
 We can interpret this problem as a game, in which for an instance $(a,b)$ players get each one part of the instance ($a$ or $b$), together with the colour of the instance,
 and each of them has to output ``I'' or ``you'', in such a way that they agree on who is the winner.

In what follows, we will consider the number of colours used by a fixed colouring function ${\cal C}$ on certain subsets of $X$. To this end, we define the following notation. For any $S \subseteq \{1,\ldots,Z\}$, define $X_S = \{(a,b)\ |\ a,b \in S, a < b\}$, and define $c_S$ to be the number of different colours used by ${\cal C}$ on the elements of $X_S$.

\begin{lemma}\label{lbcs}
Consider the pair breaking problem for parameter $Z$.
Suppose that there exists a colouring function ${\cal C} : X \longrightarrow \{1, \dots, c\}$ for which there exists a symmetry breaking function ${\cal B}$. For any positive integer $k \geq 2$ and any $S \subseteq \{1,\ldots,Z\}$, if $|S| \geq \frac{3}{2}k! + \sum_{i=0}^{k-3} \frac{k!}{(k-i)!}$, then $c_S \geq k$.
\end{lemma}
\begin{proof}
We prove the result by induction on $k$. In the base case, i.e., when $k=2$, suppose that $|S| \geq 3$, and, to obtain a contradiction, assume that $c_S < 2$. Then, for any $a,b,c \in S$ along with the single colour $\gamma$ used to colour $X_S$, we require all of the following:
\begin{enumerate}
\item ${\cal B}(a,\gamma) \neq {\cal B}(b,\gamma)$
\item ${\cal B}(a,\gamma) \neq {\cal B}(c,\gamma)$
\item ${\cal B}(b,\gamma) \neq {\cal B}(c,\gamma)$
\end{enumerate}

From (2), we see that ${\cal B}(a,\gamma) \neq {\cal B}(c,\gamma)$, so, by (3), it follows that ${\cal B}(a,\gamma) = {\cal B}(b,\gamma)$. This contradicts (1), so our assumption that $c_S < 2$ was incorrect.

As induction hypothesis, assume that, for some $k \geq 2$, if $|S| \geq \frac{3}{2}k! + \sum_{i=0}^{k-3} \frac{k!}{(k-i)!}$, then $c_S \geq k$. We now consider any set $S = \{s_1,\ldots,s_m\}$ with $m \geq \frac{3}{2}(k+1)!+ \sum_{i=0}^{k-2} \frac{(k+1)!}{(k+1-i)!}$, and prove that $c_S \geq k+1$. To obtain a contradiction, assume that $c_S \leq k$.

Consider the subset of $X_S$ consisting of $(s_i,s_m)$ for all $i \in \{1,\ldots,m-1\}$. By the Pigeonhole Principle, at least $\frac{m-1}{c_S}$ of these pairs are assigned the same colour $\alpha$ by ${\cal C}$. Let $S' = \{s_i\ |\ {\cal C}((s_i,s_m)) = \alpha\ \}$ and note that $|S'| \geq \frac{m-1}{c_S}$. We first show that $c_{S'} \geq k$, and then we will prove that this leads to a contradiction. From our assumption that $c_S \leq k$, note that $|S'| \geq \frac{m-1}{c_S} \geq \frac{m-1}{k} \geq \frac{m-1}{k+1}$. Further $m-1 \geq  \frac{3}{2}(k+1)!+ \sum_{i=0}^{k-2} \frac{(k+1)!}{(k+1-i)!} - 1 = \frac{3}{2}(k+1)!+ \sum_{i=1}^{k-2} \frac{(k+1)!}{(k+1-i)!} = \frac{3}{2}(k+1)!+ \sum_{i=0}^{k-3} \frac{(k+1)!}{(k-i)!}$. It follows that $|S'| \geq \frac{m-1}{k+1} \geq \frac{3}{2}k!+ \sum_{i=0}^{k-3} \frac{k!}{(k-i)!}$. So, by the induction hypothesis, $c_{S'} \geq k$.

Finally, we show that $c_{S'} \geq k$ contradicts our assumption that $c_S \leq k$. To do this, we first show that ${\cal C}((a,b)) \neq \alpha$ for every $(a,b) \in X_{S'}$. Consider an arbitrary $(a,b) \in X_{S'}$, and let $S'' = \{a,b,s_m\}$. Note that $X_{S''} = \{(a,b),(a,s_m),(b,s_m)\}$, and, from the base case above, note that $c_{S''} \geq 2$. Since $a,b \in S'$, it follows that $(a,s_m)$ and $(b,s_m)$ are assigned $\alpha$ by ${\cal C}$. Therefore, the remaining element of $X_{S''}$, i.e., $(a,b)$, must be assigned a colour other than $\alpha$, as desired. However, this implies that $c_S > k$. Indeed, at least $k$ colours, all different from $\alpha$ 
(which ${\cal C}$ uses to colour the pairs $(s_i,s_m)$ for each $s_i \in S'$) are used
 to colour $X_{S'}$ (see Figure \ref{grid} for an example.) This completes the induction. 
\end{proof}

\begin{figure}[!ht]
\begin{center}
\includegraphics[scale=0.8]{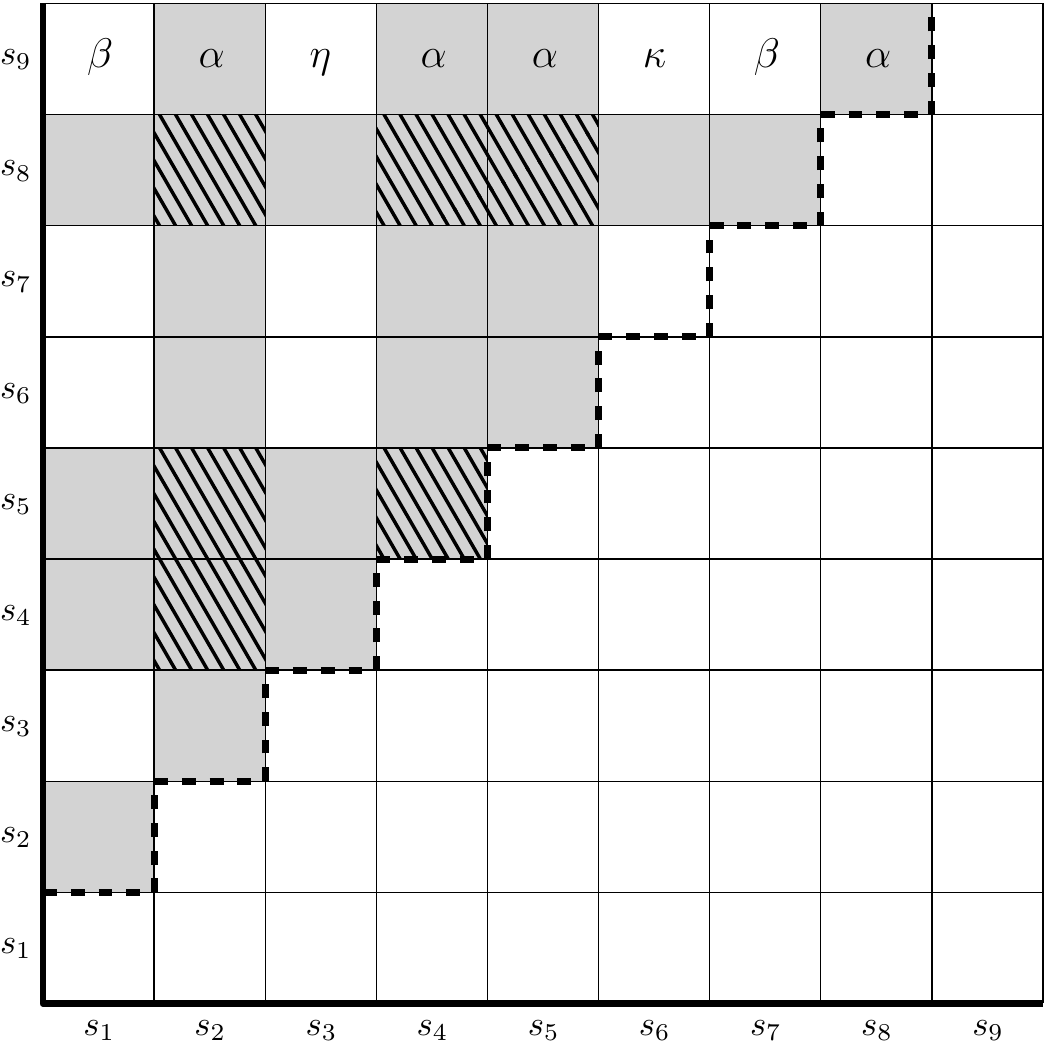}
\end{center}
\caption{An illustration of the colouring argument in Lemma \ref{lbcs}. In this example, $S = \{s_1,\ldots,s_9\}$ and $X_S$ is the section of the grid that lies above the dotted line. Assume that ${\cal C}$ has assigned an integer colour to every pair in $X_S$, and that the entries along row $s_9$ are coloured as indicated.  Then, $S' = \{s_2,s_4,s_5,s_8\}$, and we have highlighted the corresponding rows and columns. The set $X_{S'}$ consists of the entries where these rows and columns intersect, i.e., the union of the patterned entries. As demonstrated in the proof, for every $(a,b) \in X_{S'}$, entry $(a,b)$ is not coloured $\alpha$, since otherwise the three entries $(a,b)$, $(a,s_9)$, and $(b,s_9)$ would contradict the base case. Therefore, the total number of colours needed to colour $X_S$ is at least one greater than the number needed to colour $X_{S'}$.}
\label{grid}
\end{figure}

\begin{corollary}\label{lb-colour}
For any colouring function ${\cal C} : X \longrightarrow \{1, \dots ,c\}$ for which there exists a symmetry breaking function, we have $c \in \Omega(\sqrt{\log{Z}})$.
\end{corollary}
\begin{proof}
We apply Lemma \ref{lbcs} with $S = \{1,\ldots,Z\}$, with $Z \geq 3$. We set $Z = \frac{3}{2}k! + \sum_{i=0}^{k-3} \frac{k!}{(k-i)!}$ and find a lower bound for $k$ in terms of $Z$. Clearly $k \geq 2$. 
Note that $k! \leq k^k$, and that $\sum_{i=0}^{k-3} \frac{k!}{(k-i)!} = 1 + k + \ldots + \frac{k!}{3!} \leq (k-2)k^{k-3} \leq k^k$. Hence, $Z \leq 3k^k$, so $\log{Z} \leq 3k\log{k} \leq 3k^2$. Therefore, $c \geq k \geq \sqrt{\frac{1}{3}\log{Z}}$, as required.
\end{proof}

We will now use the above result on the pair breaking problem to obtain a lower bound on the size of advice for leader election in time $D-2$, for trees with odd diameter $D$.
To this end, we define the following class of trees of odd diameter $D$, called {\em double brooms}. These trees will have size $m$ which satisfies 
$\delta (\delta +1)=(m-(D-3)-2)/2$ for some integer $\delta$. Consider any such integers $m$ and $\delta$.
Let $f:  \mathbb{Z}^+ \longrightarrow
(\mathbb{Z}^+)^{\delta}$ be any (computable) bijection from positive integers to $\delta$-tuples of positive integers. 
We define the double broom $DB_{\delta}(a,b)$ of size $m$, with positive integer parameters $a<b\leq \delta^{\delta}$ as follows.  The {\em handle} of $DB_{\delta}(a,b)$ is a path of length $D-4$, with endpoints $v_a$ and $v_b$, such that the port sequence $seq(v_a,v_b)=(0,0,1,1,0,0,\dots ,1,1,0,0)$ is a palindrome. 
To each endpoint of the handle, attach the following tree of height 2 rooted at this endpoint. The first level of the tree $T_a$ attached to $v_a$ consists of $\delta +1$ nodes
$w_1,\dots ,w_{\delta +1}$. The port at $v_a$ corresponding to the edge $\{v_a, w_i\}$ is $i$ and the port at $w_i$ corresponding to this edge is 0.
Let $f(a)=(a_1,\dots a_{\delta})$ and let $a_{\delta +1}=(m-(D-3)-2)/2-\delta -\sum_{i=1}^k a_i$. (The term $a_{\delta +1}$ is defined in this way to ensure that the entire
double broom has size exactly $m$.) For each $i \in \{1,\dots ,\delta+1\}$, attach 
$a_i$ leaves to node $w_i$.  The tree $T_b$ attached to the endpoint $v_b$ is defined analogously. This concludes the description of the double broom $DB_{\delta}(a,b)$, see Fig. \ref{doublebroom}.

\begin{figure}[!ht]
\begin{center}
\includegraphics[scale=1]{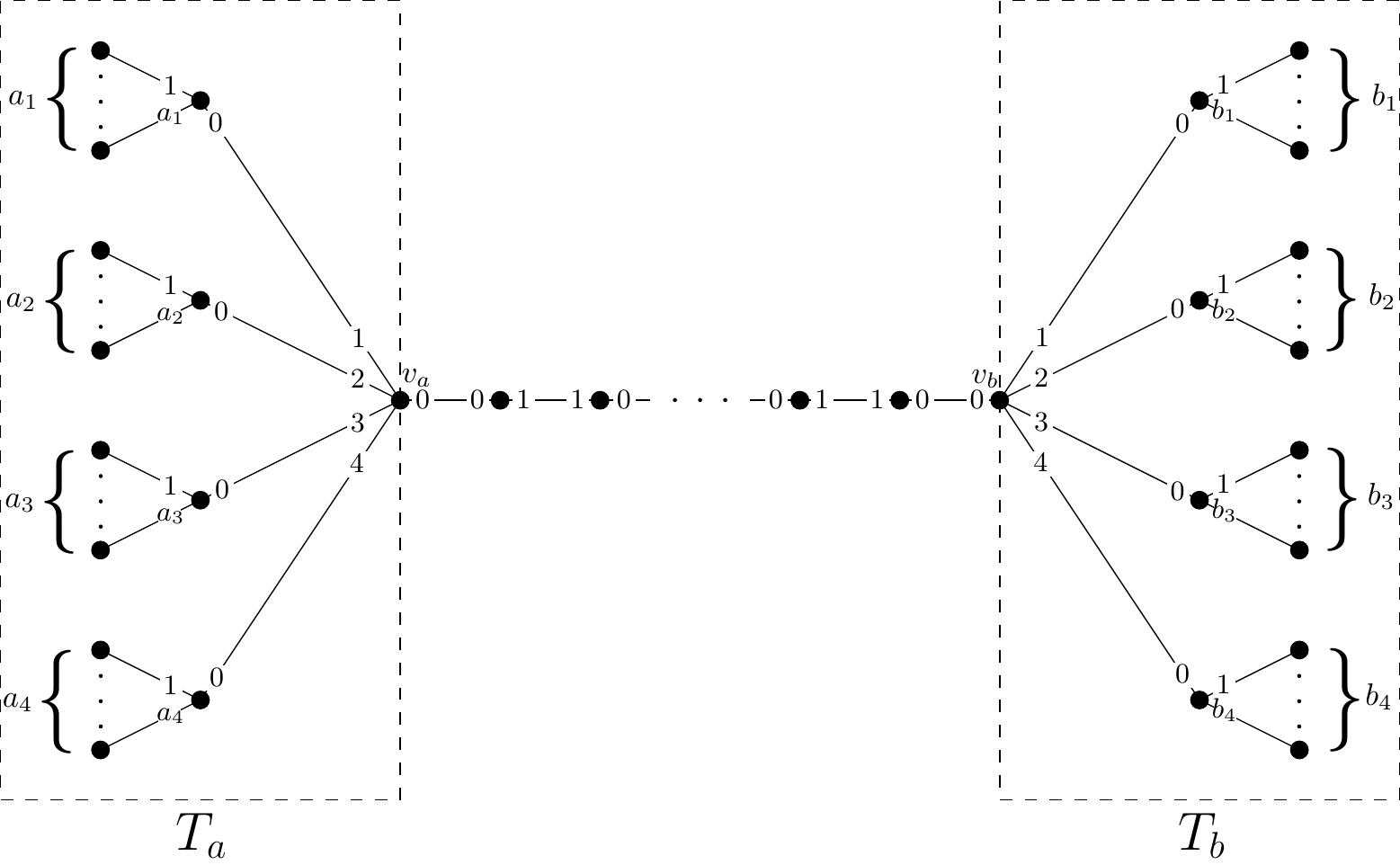}
\end{center}
\caption{Double broom $DB_{3}(a,b)$}
\label{doublebroom}
\end{figure}

The following lemma gives a reduction from the pair breaking problem to leader election in time $D-2$ for double brooms.

\begin{lemma}\label{reduct}
Suppose that there exists an algorithm $ELECT$ solving leader election in time $D-2$ for the class of double brooms $DB_{\delta}(a,b)$ with odd diameter $D$, fixed $\delta$, and all positive integer parameters $a<b\leq \delta^{\delta}$, which uses advice of size $o(\log\log \delta^{\delta})$. Then, for the pair breaking problem with
parameter $Z= \delta^{\delta}$, there exists a colouring function that uses $o(\sqrt{\log Z})$ colours for which there exists a symmetry breaking function.
\end{lemma}

\begin{proof}
We define a function $F$ which maps instances of the pair breaking problem to double brooms. More specifically, $F$ maps each pair $(a,b)$ to the double broom $DB_{\delta}(a,b)$. Let ${\cal A}$ be the advice function for algorithm $ELECT$ that maps double brooms $DB_{\delta}(a,b)$ to binary advice strings. Let $A$ be the the range of ${\cal A}$, i.e., the set of all advice strings needed by the algorithm. Let $g : A \longrightarrow \{1,\ldots,|A|\}$ be any bijection from binary strings to positive integers.

For the pair breaking problem with parameter $Z$, we define the colouring function ${\cal C} : X \longrightarrow \{1,\ldots,c\}$ that takes each instance $(a,b)$ of the pair breaking problem and maps it to $g({\cal A}(F(a,b)))$. Intuitively, the colour of an instance $(a,b)$ is set to the advice given for the corresponding double broom. 
We now show that $c \in o(\sqrt{\log Z})$. Since the size of advice for algorithm $ELECT$ is $o(\log\log \delta^{\delta})$, it follows that every advice string provided to the algorithm has size less than $\frac{1}{4}\log\log \delta^{\delta}$, for sufficiently large $\delta$. It follows that the number of different advice strings needed by the algorithm is at most $(\log \delta^\delta)^{1/4}$. Therefore, the range of $g$ has size at most $(\log \delta^\delta)^{1/4} \in o(\sqrt{\log Z})$, as required.

Next, we show that there is a symmetry breaking function ${\cal B}$ that uses this colouring function. For any integer $z \in \{1,\ldots,Z\}$, we map $z$ to a double broom with one parameter equal to $z$ and the other parameter equal to the smallest positive integer not equal to $z$. In particular, let $G$ be a function such that $G(1) = DB_{\delta}(1,2)$ and $G(z) = DB_{\delta}(1,z)$ for all $z \in \{2,\ldots,Z\}$. Define the binary function ${\cal B}$ as follows. It takes integer inputs $z,\gamma$, runs the ELECT algorithm on $G(z)$ with advice $g^{-1}(\gamma)$, and outputs 0 if and only if the elected node is within distance $(D-5)/2$ from $v_z$ (i.e., closer to the endpoint of the handle to which $T_z$ is attached). We now prove that ${\cal B}$ is indeed a symmetry breaking function. For any instance $(a,b)$ of the pair breaking problem, consider the values of ${\cal B}(a,{\cal C}(a,b))$ and ${\cal B}(b,{\cal C}(a,b))$. First, note that $g^{-1}({\cal C}(a,b))$ is the advice string, say $s$, that is provided to the $ELECT$ algorithm for the tree $DB_{\delta}(a,b)$. So, on the input pair $(a,{\cal C}(a,b))$, our function ${\cal B}$ runs $ELECT$ on $G(a)$ with advice $s$. Note that, in the construction of both $G(a)$ and $DB_{\delta}(a,b)$, the same tree $T_a$ is attached to node $v_a$. Since the $ELECT$ algorithm uses time $D-2$, it follows that, for any leaf $v$ in $T_a$, we have  $V_{G(a)}(v,D-2)=V_{DB_{\delta}(a,b)}(v,D-2)$  (since the handle has length $D-4$ and $T_a$ has height 2). Hence, $v$ elects the same node in $G(a)$ as it does when $ELECT$ is executed on $DB_{\delta}(a,b)$. Similarly, on the input pair $(b,{\cal C}(a,b))$, the function ${\cal B}$ runs $ELECT$ on $G(b)$ with advice $s$, and elects the same node in $G(b)$ as it does when $ELECT$ is executed on $DB_{\delta}(a,b)$. However, for any node $w$ in $DB_{\delta}(a,b)$, exactly one of $v_a$ or $v_b$ is within distance $(D-5)/2$ from $w$. Therefore, ${\cal B}(a,{\cal C}(a,b)) \neq {\cal B}(b,{\cal C}(a,b))$, and hence $\cal B$ is symmetry breaking, as required.
\end{proof}

Corollary \ref{lb-colour} and Lemma \ref{reduct} imply the following theorem.

\begin{theorem}\label{lb-log}
Let $D<n$ be positive integers, where $D$ is odd. There exists a class $\cal T$ of trees $T$ with size $\Theta(n)$, diameter $D$ and $\xi(T) \leq D-2$, such that
every leader election algorithm working in time $D-2$ on the class $\cal T$ requires advice of size $\Omega(\log n)$.
\end{theorem}

\begin{proof}
Choose integers $m$ and $\delta$, such that $\delta (\delta +1)=(m-(D-3)-2)/2$  and $\delta \in \Theta(\sqrt{n})$. Hence $m \in \Theta(n)$.
Let $\cal T$ be the class of double brooms $DB_{\delta}(a,b)$, for all positive integers $a<b\leq \delta^{\delta}$.
Corollary \ref{lb-colour} and Lemma \ref{reduct} imply that the size of advice required by any leader election algorithm on this class is
$\Omega(\log\log \delta^{\delta}) \subseteq \Omega(\log n)$.

For each node $v$ in any $DB_{\delta}(a,b)$, the view $V_{DB_{\delta}(a,b)}(v,D-2)$ includes both $v_a$ and $v_b$, as well as at least one of the subtrees $T_a$ or $T_b$. Therefore, 
given the entire map of  $DB_{\delta}(a,b)$, node $v$ can locate itself on the map. Further, node $v$ can 
compute values of $a$ and $b$ and output the sequence of outgoing ports of $seq(v,v_{\mu})$, where
$\mu=\min\{a,b\}$. Hence $\xi(T) \leq D-2$.
\end{proof}

Theorems \ref{ub-log} and \ref{lb-log} imply the following corollary.

\begin{corollary}
The minimum size of advice sufficient to do leader election in time $D-2$ in the class of $\Theta(n)$-node trees $T$ with odd diameter $D$ and  $\xi(T) \leq D-2$ is $\Theta(\log n)$.
\end{corollary}

\section{Time $\beta \cdot diam \leq \tau \leq diam -3$, for any constant $\beta >1/2$}

For the time interval $[\beta \cdot diam,   diam -3]$, for any constant $\beta >1/2$, we provide upper and lower bounds on the minimum size of advice sufficient to perform leader election. Our bounds are separated by a gap of only $O(\log  n)$, except for the special case when $diam$ is even and time is $diam-3$.

\subsection{Upper bound}

Consider a tree $T$ with diameter $diam = D$, and let $h = \lfloor D/2 \rfloor$.
Let $\epsilon = \frac{\tau}{D} - \frac{1}{2}$, and note that $\epsilon < \frac{1}{2}$.

At a high level, our algorithm first partitions the set of leaves into $k$ classes such that leaves in the same class belong to the same subtree of height $\lfloor \epsilon D \rfloor - 1$. Let $R_1,\ldots,R_k$ be the roots of these subtrees of height $\lfloor \epsilon D \rfloor-1$. Next, each node in $T$ chooses one of these $R_i$ as its \emph{representative}. The advice provided to the algorithm consists of $k$ pieces, one piece for each $R_i \in \{R_1,\ldots,R_k\}$. Each piece instructs how to reach the leader starting at node $R_i$. Therefore, to solve leader election, each node $v$ can compute a path to the leader using the path to its representative, along with the advice. The main difficulty in designing the algorithm is to ensure that each node finds its representative and determines which piece of advice corresponds to it.

In what follows, we will make use of an injective function $F$ that maps rooted trees with at most $n$ nodes to binary strings of fixed length in $O(n)$. (One example of such a function is discussed in Section \ref{alphadiam}.) We will apply this function to views of nodes in order to distinguish them. 
When $F$ is applied to a node's full view of the tree (i.e., to $V(v,n)$), then the resulting binary string will be called the node's \emph{signature}. Given two distinct nodes $v,w \in T$, we will say that \emph{$v$ has a smaller signature than $w$} if $F(V(v,n))$ is lexicographically smaller than $F(V(w,n))$.

To aid in the description and analysis of our algorithm, we carefully choose a node $c$  as the root  of $T$. This is the node that the algorithm will elect. In the case where $D$ is even, the central node of $T$ is chosen as the root. In the case where $D$ is odd, the node on the central edge of $T$ that has the smaller signature is chosen as the root. In what follows, the depth of a node is defined as its distance from the root $c$.

\paragraph{The Representatives.} We define the \emph{representatives} of an arbitrary tree $K$ using the following greedy subroutine that takes as input the map of $K$ and a designated root $r$ on the map. 

\begin{algorithm}[H]
\caption{\texttt{ComputeReps}($K,r$)}
\begin{algorithmic}[1]
\State $i \leftarrow 0$
\State $UncoveredLeaves \leftarrow $ set of all leaves of $K$ that have depth at least $\lfloor \epsilon D \rfloor-1$
\State {\bf While} $UncoveredLeaves \neq \emptyset$
\State \indent $i \leftarrow i+1$
\State \indent $x_i \leftarrow$ deepest leaf in $UncoveredLeaves$
\State \indent $r_i \leftarrow $ ancestor of $x_i$ at distance $\lfloor \epsilon D \rfloor - 1$
\State \indent Remove from $UncoveredLeaves$ any descendants of $r_i$
\State Output $(r_1,\ldots,r_i)$
\end{algorithmic}
\label{computereps}
\end{algorithm}


Since the size of $UncoveredLeaves$ decreases by at least one in every iteration of the while loop, there exists some positive integer $k$ such that the procedure terminates after $k$ iterations of the while loop. Note that $path(x_1,r_1),\ldots,path(x_k,r_k)$ are $k$ disjoint paths, each containing $\lfloor \epsilon D \rfloor$ nodes of $K$. It follows that $k \in O(|K|/D)$.

\paragraph{The Advice.} We now describe the advice provided to the algorithm for tree $T$. First, the oracle computes the  representatives $R_1,\ldots,R_k$ of $T$ by executing \texttt{ComputeReps}($T,c$). Then, for each $i \in \{1,\ldots,k\}$, the oracle computes a list $L_i$ consisting of the sequences $seq(R_i,w)$ for every $w$ at distance at most $h$ from $R_i$. 
 Each $L_i$ is sorted lexicographically, and $z_i$ is defined to be the index of $seq(R_i,c)$ in this list. The purpose of the integers $z_1,\ldots,z_k$ is to enable each $R_i$ (and each node that has representative $R_i$) to compute the path from $R_i$ to $c$. However, we cannot assume that nodes know the index of their representative, i.e., nodes may not know which $z_i$ should be used to compute the path from their representative to $c$. To remedy this, the oracle includes in the advice a \emph{trie} \cite{AHU} that enables each node to determine which piece $z_i$ of advice is intended
 for it to use. More specifically, the oracle first computes a list $S$ consisting of $F(V(R_i,h))$ for each $i \in \{1,\ldots,k\}$. Then, a trie is computed for $S$ using the following recursive procedure \texttt{BuildTrie}.

\begin{algorithm}[H]
\caption{\texttt{BuildTrie}($S$)}
\begin{algorithmic}[1]
\State {\bf If} $S$ contains only one string $s = F(V(R_i,h))$ for some $i \in \{1,\ldots,k\}$:
\State \indent Return a single node labeled $z_i$
\State {\bf Else}:
\State \indent 	$j \leftarrow$ the largest index such that all strings in $S$ have the same prefix of length $j$
\State \indent	$S_0 \leftarrow$ the list of strings in $S$ that have a 0 at index $j+1$, each with prefix of length $j+1$ removed
\State \indent	$S_1 \leftarrow$ the list of strings in $S$ that have a 1 at index $j+1$, each with prefix of length $j+1$ removed
\State \indent	Return a node labeled $j$ with left child equal to $\mathtt{BuildTrie}(S_0)$ and with right child equal to $\mathtt{BuildTrie}(S_1)$
\end{algorithmic}
\label{buildtrie}
\end{algorithm}


The advice provided to the algorithm is the value of $D$, the value of the allocated time $\tau$, as well as the trie computed by \texttt{BuildTrie}($S$). It remains to show that the advice is well-defined, i.e., that \texttt{BuildTrie}($S$) terminates. This is the case if and only if the strings in $S$ are all distinct.
So, the following lemma proves that  \texttt{BuildTrie}($S$) terminates for all trees in which leader election is solvable in time $\tau$. 

\begin{lemma}
If  $\xi(T) \leq \tau$ then, for all distinct $i,j \in \{1,\ldots,k\}$, we have $F(V(R_i,\tau)) \neq F(V(R_j,\tau))$.
\end{lemma}
\begin{proof}
Suppose that there is an algorithm $ELECT$ that solves leader election in time $\tau$ (with any amount of advice). To obtain a contradiction, assume that, for some distinct $i,j \in \{1,\ldots,k\}$, we have $F(V(R_i,\tau)) = F(V(R_j,\tau))$. Since $F$ is injective, it follows that $V(R_i,\tau) = V(R_j,\tau)$. This implies that nodes $R_i$ and $R_j$ output the same value in the execution of $ELECT$ on $T$. Let $P_i$ be the simple path in $T$ starting at $R_i$ that corresponds to $R_i$'s output, and let $P_j$ be the simple path in $T$ starting at $R_j$ that corresponds to $R_j$'s output. Note that $|P_i| = |P_j|$. Since $R_i$ and $R_j$ must elect the same node, there exists a unique node $v$ such that the paths $P_i$ and $P_j$ first intersect at $v$. Since $T$ is a tree, we can re-write $P_i = path(R_i,v) \cdot Q$ and $P_j = path(R_j,v) \cdot Q$ for some (possibly empty) simple path $Q$. Since $|P_i| = |P_j|$, it follows that $|path(R_i,v)| = |path(R_j,v)|$. We now observe that $|path(R_i,v)|$ and $|path(R_j,v)|$ are bounded above by $h$. Indeed, if both paths had length greater than $h$, then their union would be a simple path of length at least $2h+2 > D$, a contradiction. So, $path(R_i,v)$ and $path(R_j,v)$ are edge disjoint paths of length at most $h$. However, this means that the port at $v$ corresponding to the final edge of $path(R_i,v)$ is different than the port at $v$ corresponding to the final edge of $path(R_j,v)$. It follows that $seq(R_i,v) \neq seq(R_j,v)$. This implies that $V(R_i,h) \neq V(R_j,h)$, and, since $h \leq \tau$, it follows that $V(R_i,\tau) \neq V(R_j,\tau)$, a contradiction.
\end{proof}

We also provide a retrieval procedure $\mathtt{Retrieve}(TR,s)$ which takes as input a trie $TR$ and a string $s$. If $s$ belongs to the set $S$ used in the construction of the trie, then the procedure returns the value stored in the label of the leaf node corresponding to the string $s$. In our case, each $s$ is some $F(V(R_i,h))$, and the value stored in the corresponding leaf is $z_i$. 


\begin{algorithm}[H]
\caption{\texttt{Retrieve}($TR,s$)}
\begin{algorithmic}[1]
\State $root \leftarrow$ root of $TR$
\State {\bf If} $TR$ consists of a single node:
\State \indent Return label of $root$
\State {\bf Else}:
\State \indent $j \leftarrow$ label of $root$
\State \indent {\bf If} $(j+1)^{th}$ bit of $s$ is 0:
\State \indent \indent $TR' \leftarrow$ left subtree of $root$
\State \indent {\bf Else}:
\State \indent \indent $TR' \leftarrow$ right subtree of $root$
\State \indent $s' \leftarrow$ string $s$ with prefix of length $j+1$ removed
\State \indent Return \texttt{Retrieve}($TR',s'$)
\end{algorithmic}
\label{retrieve}
\end{algorithm}

\paragraph{The Algorithm.} We now define our leader election algorithm \texttt{ElectWithTrie} executed by each node $v$ in $T$ given an advice string $A$. We start by defining a procedure \texttt{FindRep}. For each node $v$ with sufficiently large depth, \texttt{FindRep} computes the representative of $v$ given its view $V(v,\tau)$ and
the values of $\epsilon$ and $D$. 
(Nodes with small depth will not need representatives to perform election.)
At a high level, the procedure picks an ancestor $w$ of $v$, finds the subtree rooted at $w$ consisting of all of $w$'s descendants, then executes \texttt{ComputeReps} on this tree. Of the representatives returned by \texttt{ComputeReps}, $v$ picks one that is either its descendant or ancestor.

\begin{algorithm}[H]
\caption{\texttt{FindRep}($v,V(v,\tau),\epsilon,D$)}
\begin{algorithmic}[1]
\State $w \leftarrow$ the node in $V(v,\tau)$ at distance $\lfloor \epsilon D \rfloor-1$ from $v$ such that every endless path in $V(v,\tau)$ starting at $v$ passes through $w$ \label{definew}
\State $up_w \leftarrow$ the edge incident to $w$ s.t. every endless path starting at $v$ passing through $w$ uses $up_w$
\State $T_w \leftarrow $ subtree of $V(v,\tau)$ induced by all nodes $x$ such that $path(x,v)$ does not pass through $up_w$
\State $R \leftarrow \mathtt{ComputeReps}(T_w,w)$ \label{vComputeReps}
\State $\ell \leftarrow$ any leaf in $V(v,\tau)$ such that $v$ lies on $path(\ell,w)$ \label{defineleaf}
\State Return any node $r$ in $R$ such that $r$ lies on $path(\ell,w)$ \label{definerep}
\end{algorithmic}
\label{findrep}
\end{algorithm}

Below is the pseudocode of the algorithm \texttt{ElectWithTrie} that is executed by each node $v$ in $T$.

\begin{algorithm}[H]
\caption{\texttt{ElectWithTrie}($A$)}
\begin{algorithmic}[1]
\State $D \leftarrow$ diameter of $T$ provided in $A$
\State $h \leftarrow \lfloor D/2 \rfloor$
\State $\tau \leftarrow$ the value of allowed time provided in $A$
\State Use $\tau$ rounds to learn $V(v,\tau)$
\State {\bf If} $V(v,\tau)$ contains no endless paths starting at $v$: \label{smalldepth}
\State \indent {\bf If} $D$ is even:
\State \indent \indent $c \leftarrow$ the central node of $V(v,\tau)$
\State \indent {\bf Else}:
\State \indent \indent $c \leftarrow$ the node on the central edge of $V(v,\tau)$ with smaller signature
\State \indent Output sequence of outgoing ports in $seq(v,c)$ \label{shallowelect}
\State {\bf Else}:
\State \indent $\epsilon \leftarrow\frac{\tau}{D} -\frac{1}{2}$
\State \indent $r \leftarrow \mathtt{FindRep}(v, V(v,\tau),\epsilon,D)$ \label{pickrep}
\State \indent $s \leftarrow F(V(r,h))$ \label{getsig}
\State \indent $TR \leftarrow$ the trie provided in $A$
\State \indent $z \leftarrow \mathtt{Retrieve}(s,TR)$ \label{definez}
\State \indent $L \leftarrow$ lexicographically-ordered list of port sequences $seq(r,w)$ where $d(r,w) \leq h$ \label{defineL}
\State \indent $P \leftarrow$ path corresponding to $z^{th}$ sequence in $L$ \Comment{path from $r$ to $c$} \label{getpath}
\State \indent $W \leftarrow$ the walk in $T$ consisting of $path(v,r)$ followed by $P$
\State \indent $Q \leftarrow$ simple path from $v$ to $c$ obtained from $W$ by removing any non-simple subwalk
\State \indent Output sequence of outgoing ports obtained from $Q$'s port sequence
\end{algorithmic}
\label{ElectWithDT}
\end{algorithm}

\begin{theorem}\label{ub-DT}
Consider any fixed $D, \beta$ such that $\beta > 1/2$ and $\beta D \leq D-3$. For $\tau \in [ \beta D, D-3]$, Algorithm {\tt ElectWithTrie} solves leader election in trees $T$ with size $n$, diameter $D$ and $\xi(T)\leq \tau$, in time $\tau$ using $O(\frac{n\log{n}}{D})$ bits of advice.
\end{theorem}
\begin{proof}
For any node $v$, denote by $depth(v)$ the depth of $v$ with respect to $c$. To prove the correctness of our algorithm, we show that every node elects $c$. We first consider the nodes that have small depth.

\begin{claim}
For any node $v$ with depth less than $\lfloor \epsilon D \rfloor$, $v$ elects $c$ at line \ref{shallowelect}.
\end{claim}
To prove the claim, note that $\tau \geq \frac{1}{2}D + \lfloor \epsilon D \rfloor \geq \frac{1}{2}D + (depth(v) + 1) \geq h+1+depth(v)$. Since the distance from $c$ to any node in $T$ is at most $h+1$, it follows that the distance from $v$ to any node in $T$ is at most $h+1+depth(v)$. Therefore, $V(v,\tau)$ is equal to $T$. Hence, $V(v,\tau)$ contains no endless paths starting at $v$, and, at line \ref{shallowelect}, $v$ elects $c$. This proves the claim.

In what follows, we consider the nodes with large depth, i.e., nodes $v$ with $depth(v) \geq \lfloor \epsilon D \rfloor$. We first show that each such node picks one of the representatives $R_1,\ldots,R_k$ at line \ref{pickrep}.

\begin{claim}\label{repsubset}
$\mathtt{FindRep}(v, V(v,\tau),\epsilon,D)$ outputs a node in $\{R_1,\ldots,R_k\}$.
\end{claim}
To prove the claim, we show that the set $R$ computed at line \ref{vComputeReps} of Algorithm \texttt{FindRep} is a subset of $\{R_1,\ldots,R_k\}$, i.e., a subset of the output of \texttt{ComputeReps}$(T,c)$. The proof proceeds in three steps. First, we prove that, in the execution of Algorithm \texttt{FindRep}$(v, V(v,\tau),\epsilon,D)$, $w$ is an ancestor of $v$ in $T$ such that $w$ has at least one descendant in $T$ at distance $\lfloor \epsilon D \rfloor - 1$. Next, we show that $T_w$ is induced by the descendants of $w$. Finally, we show that, for such $w$ and $T_w$, we have $\mathtt{ComputeReps}(T_w,w) \subseteq \mathtt{ComputeReps}(T,c)$.

{\bf Showing that $w$ is an ancestor of $v$ and $w$ has at least one descendant in $T$ at distance $\lfloor \epsilon D \rfloor - 1$.}  We consider $w$'s definition at line \ref{definew} of \texttt{FindRep}. Recall that $v$ has depth at least $\lfloor \epsilon D \rfloor$ in $T$, so $v$ has an ancestor $a$ at distance $\lfloor \epsilon D \rfloor-1$. At least one endless path in $V(v,\tau)$ starting at $v$ passes through $a$ since there is at least one endless path starting at $v$ that passes through $c$ (which is an ancestor of $a$). Hence, $a$ is a node that satisfies the definition on line \ref{definew}. Moreover, $a$ is the only such node, since we can show that every path in $V(v,\tau)$ starting at $v$ that avoids $a$ is a terminated path. Indeed, consider any path from $v$ to a leaf $b$ such that the path avoids $a$. The distance from $v$ to the penultimate node on $path(v,a)$ is $\lfloor \epsilon D \rfloor-2$ and the distance from this node to any of its descendants is at most $h+1$. Therefore, $d(v,b) \leq h+\lfloor \epsilon D \rfloor - 1 < (\frac{1}{2} + \epsilon)D \leq \tau$. Thus, we have shown that $a$ is the unique node that satisfies the definition of $w$, so $w$ is an ancestor of $v$. Also, by definition, $d(v,w) = \lfloor \epsilon D \rfloor-1$, so $w$ has at least one descendant in $T$ at distance $\lfloor \epsilon D \rfloor - 1$.

{\bf Showing that  $T_w$ is induced by the descendants of $w$ in $T$.}  We show that $up_w$ is the edge between $w$ and $w$'s parent. As we observed above, there is at least one endless path in $V(v,\tau)$ starting at $v$ that passes through $w$, so there is at least one edge that satisfies the definition of $up_w$. Moreover, since the distance from $w$ to any of its descendants $b$ is at most $h+1$, and the distance from $v$ to $w$ is $\lfloor \epsilon D \rfloor - 1$, it follows that $d(v,b) \leq h + \lfloor \epsilon D \rfloor \leq (\frac{1}{2} + \epsilon)D \leq \tau$, so no endless path starting at $v$ has descendants of $w$ as both of its endpoints. It follows that $up_w$ cannot be an edge on a path from $w$ to a descendant of $w$, and, that every descendant of $w$ is in $V(v,\tau)$. This implies that $T_w$ is induced by the descendants of $w$ in $T$.

{\bf Showing that $\mathtt{ComputeReps}(T_w,w) \subseteq \mathtt{ComputeReps}(T,c)$.} The proof is by contradiction.
First note that, for any node $x$ in $T_w$, $d(x,w)$ is bounded above by $d(x,c)$. Next, to obtain a contradiction, consider the first iteration $j$ of $\mathtt{ComputeReps}(T_w,w)$ during which a node $r_j$ is chosen as one of the outputs and $r_j \not\in \{R_1,\ldots,R_k\}$. By the specification of $\mathtt{ComputeReps}$, $r_j$ was added to the output because it was the ancestor of the deepest node $x_j$ in $UncoveredLeaves$, and, further, we have $d(r_j,x_j) = \lfloor \epsilon D \rfloor - 1$. Note that $d(x_j,w)\geq \lfloor \epsilon D \rfloor - 1$, so  $d(x_j,c)\geq \lfloor \epsilon D \rfloor - 1$. We now consider the execution of $\mathtt{ComputeReps}(T,c)$. Since $d(x_j,c)\geq \lfloor \epsilon D \rfloor - 1$, it follows that $x_j$ is initially in $UncoveredLeaves$. This means that, in some iteration of the while loop, $x_j$ is removed from $UncoveredLeaves$ as the descendant of some chosen $R_i \neq r_j$. 
Since $R_i \neq r_j$ and  both $r_j$ and $R_i$ are ancestors of $x_j$, it follows that either $R_i$ is an ancestor of $r_j$ or vice-versa.
We now show that  $R_i$ is not an ancestor of $r_j$.
If this were the case, then $d(R_i,x_j) > \lfloor \epsilon D \rfloor - 1$. In particular, $x_j$ would be a leaf in $T$ such that $x_j$'s distance to $R_i$ is greater than $\lfloor \epsilon D \rfloor - 1$, which contradicts the choice of $R_i$ as the ancestor at distance exactly $\lfloor \epsilon D \rfloor - 1$ from the deepest leaf $x_i$ in $UncoveredLeaves$. So, we have that $r_j$ is an ancestor of $R_i$. However, this means that $d(r_j,x_i) > d(R_i,x_i) = \lfloor \epsilon D \rfloor - 1 = d(r_j,x_j)$, which we will use later to obtain the desired contradiction. We now reconsider the $j^{th}$ iteration of execution $\mathtt{ComputeReps}(T_w,w)$. Note that $R_i$ was not added to the output before this iteration (since, as $R_i$ is an ancestor of $x_j$, this would imply that $x_j$ was already removed from $UncoveredLeaves$, contradicting our choice of $x_j$.) By assumption, $r_i \in \{R_1,\ldots,R_k\}$ in all iterations $i < j$, so it follows that $x_i$ is in $UncoveredLeaves$ at the start of iteration $j$. But, recall that $d(r_j,x_i) > d(r_j,x_j)$, so we have a leaf with depth greater than $x_j$ in $UncoveredLeaves$ at the start of iteration $j$, which contradicts the definition of $x_j$. This concludes the proof that $\mathtt{ComputeReps}(T_w,w) \subseteq \mathtt{ComputeReps}(T,c)$, which completes the proof of the claim.

Let $R_i$ be the representative picked by $v$. It remains to show that $v$ computes $seq(R_i,c)$ in lines \ref{getsig}-\ref{getpath} of \texttt{ElectWithTrie}. First, we show that $v$ is able to compute a sufficiently large part of $R_i$'s view.

\begin{claim}\label{repview}
$V(R_i,h) \subseteq V(v,\tau)$
\end{claim}
To prove the claim, we first show that $v$ is either a descendant of,  an ancestor of, or equal to $R_i$. This is the case since, by lines \ref{defineleaf} and \ref{definerep} of \texttt{FindRep}, there is a leaf $\ell$ such that $v$ is an ancestor of $\ell$, and $R_i$ lies on a path from $\ell$ to an ancestor of $v$. 

Next, we show that $V(R_i,h) \subseteq V(v,\tau)$ regardless of whether $v$ is a descendant or ancestor of $R_i$. (The claim is obvious for $v=R_i$.)
If $v$ is a descendant of $R_i$, we note that, by the definition of $R_i$ in \texttt{ComputeReps}, the distance from $R_i$ to any of its descendants is at most $\lfloor \epsilon D \rfloor$. It follows that $\tau = (\frac{1}{2}D + \epsilon)D \geq h + d(v,R_i)$, so $V(v,\tau)$ contains $V(R_i,h)$, as desired.
If $v$ is an ancestor of $R_i$, we note that $V(v,h)$ contains all descendants of $v$. The only nodes in $V(R_i,h)$ that are not descendants of $v$ are contained in $V(v,h-d(R_i,v)) \subseteq V(v,h)$, as desired. This completes the proof of the claim.

We now prove that $v$ correctly computes the sequence of ports from its representative to $c$.

\begin{claim}\label{goodchoice}
At line \ref{getpath}, the $z^{th}$ sequence in $L$ is equal to $seq(R_i,c)$.
\end{claim}
To prove the claim, note that, in the advice construction, the $(z_i)^{th}$ sequence in $L_i$ is equal to $seq(R_i,c)$. We show that $z=z_i$ and $L=L_i$.

To prove that $z=z_i$, note that $z$ is assigned the output of $\mathtt{Retrieve}(s,TR)$, where $s = F(V(R_i,h))$ is one of the strings in $S$ used to build $TR$. It follows that $\mathtt{Retrieve}(s,TR)$ returns~$z_i$.

To prove that $L=L_i$, note that, on line \ref{defineL}, $L$ is defined as the lexicographically-sorted list of sequences $seq(r,w)$ for all $w$ with $d(r,w) \leq h$. Since $r = R_i$, this matches the definition of $L_i$. Further, by Claim \ref{repview}, we have that $V(R_i,h) \subseteq V(v,\tau)$, so the computation of $L$ can indeed be carried out by $v$. This concludes the proof of the claim.

By Claims \ref{repsubset} and \ref{goodchoice} every node with depth at least $ \lfloor \epsilon D \rfloor$ chooses a representative $R_i$ within its view
and computes a path from $R_i$ to $c$. Hence it computes a path from itself to $c$. This concludes the proof of correctness.

Finally, we consider the size of the advice. In the advice construction, the list $S$  consists of $k$ strings (one for each representative), and these strings have some fixed length, say $\lambda$, in $O(n)$ (by our choice of $F$.) Consider the trie $TR$ constructed by \texttt{BuildTrie}$(S)$, as described in Algorithm \ref{buildtrie}. 

\begin{claim}\label{boundleaves}
The number of leaves in $TR$ is at most $k$.
\end{claim}
To prove the claim, we provide a one-to-one correspondence $f$ from the leaves of $TR$ to the strings in $S$. First, for each node $w \in TR$, let $S_w \subseteq S$ be the list of strings that was provided as the parameter to \texttt{BuildTrie} in the execution where $w$ was created. Let $j_w$ be the label of node $w$, and let $pre_w$ be the common prefix of length $j_w$ of all strings in $S_w$ (if $S_w$ contains only one string $s$, then $pre_w$ is defined to be $s$.) Next, for an arbitrary leaf $x \in T$, let $(x_1,\ldots,x_m)$ be the root-to-leaf path of vertices ending at $x$ (i.e., $x_1$ is the root of $TR$ and $x_m=x$.) For every $\alpha \in \{1,\ldots,m-1\}$, we define $b_{x,\alpha}$ to be 0 if $x_{i+1}$ is the left child of $x_i$, and 1 if $x_{i+1}$ is the right child of $x_i$. Finally, we define $f(x) = pre_{x_1} \cdot b_{x,1} \cdot pre_{x_2} \cdot b_{x,2} \cdots b_{x,m-1} \cdot pre_{x_m}$, where $\cdot$ is the string concatenation operator. To see why $f$ is one-to-one, consider any distinct leaves $x,y \in TR$, and let $x_a = y_a=z$ be their deepest common ancestor. By the maximality of $a$, leaves  $x$ and $y$ are descendants of different children of $z$, so we have $b_{x,c} \neq b_{y,c}$. It follows that the bit at position $(j_{v_1} + 1) +( j_{v_2} + 1) + \cdots +( j_{v_c} + 1)$ in $f(x)$ differs from the bit at the same position in $f(y)$, so $f(x) \neq f(y)$. This concludes the proof of the claim.

\begin{claim}\label{numnodes}
The number of nodes in $TR$ is at most $2k$.
\end{claim}
To prove the claim, observe that every node in the trie is either a leaf or has two children. Consider the mapping $g$ that maps each internal node $w$ to the rightmost leaf in $w$'s left subtree. Since $g$ is one-to-one, we get that the number of internal nodes is bounded above by the number of leaves. Therefore, by Claim \ref{boundleaves}, the number of nodes in $TR$ is at most $2k$, which proves the claim.

\begin{claim}\label{labelsize}
The label of each node in $TR$ has size $O(\log{n})$.
\end{claim}
To prove the claim, note that each internal node of $TR$ is labeled with an integer corresponding to an index within a string of length at most $\lambda \in O(n)$. Therefore, the size of each such label is $O(\log{n})$. Next, each leaf is labeled with an integer $z_i$ corresponding to an index within a list of sequences, all of which correspond to simple paths in $T$ originating at representative $R_i$. Since the number of such simple paths is bounded above by $n-1$ (one for each node in $T$ other than $R_i$) it follows that $z_i \leq n-1$. Therefore, the size of $z_i$ is $O(\log{n})$. This concludes the proof of the claim.

By Claims \ref{numnodes} and \ref{labelsize}, the total number of bits needed to represent $TR$ is $O(k\log{n})$. Recall that $k \in O(n/D)$, so $TR$ can be represented using $O(\frac{n\log{n}}{D})$ bits. Since providing the diameter of $T$ and the value of $\tau$ require only $O(\log{n})$ additional bits, we are done.
\end{proof}

\subsection{Lower bound}

The lower bound holds even for a slightly larger time interval than we need, namely starting from $\lfloor diam /2 \rfloor $.
We split the argument into two cases: when the diameter is odd and when it is even.

\subsubsection{Odd diameter}

\begin{theorem}\label{lb2odd}
Let $7 \leq D<n$ be positive integers, where $D$ is odd. 
Fix any value $\lfloor \frac{D}{2} \rfloor \leq \tau \leq D -3$.
There exists a class $\cal T$ of trees $T$ with size $\Theta(n)$, diameter $D$, and $\xi(T) \leq \tau$, such that
every leader election algorithm working in time $\tau$ on the class $\cal T$ requires advice of size $\Omega(n/D)$.
\end{theorem}
\begin{proof}
Let $h = \lfloor \frac{D}{2} \rfloor$ and let $k  = \left\lceil \frac{n}{D} \right\rceil$. Let $m = Dk + 2 \in \Theta(n)$. We define a class of trees $T$ with size $m \in \Theta(n)$,
odd diameter $D$, and  $\xi(T) \leq \lfloor \frac{D}{2} \rfloor $ such that the minimum size of advice needed by an arbitrary algorithm $ELECT$ solving leader election in time $\tau$  for this class is $\Omega(n/D)$.

We start with a single tree $G$ of size $m$, defined as follows. The edge $\{c_0,c_1\}$ is the central edge of $G$, and the port numbers corresponding to this edge at $c_0$ and $c_1$ are both 0. Next, for each $i \in \{1,\ldots,k\}$, there is a path $P_i$ of length $h-2$ with $c_0$ as one endpoint. The other endpoint of each of these paths will be denoted by $p_i$. Further, the port sequence $seq(c_0,p_i)$ is equal to $(i,0,1,0,1,0,\ldots,1,0)$. The same paths appear with $c_1$ as one endpoint, and, for each $i \in \{1,\ldots,k\}$, we will refer to each of these paths, and their corresponding endpoint other than $c_1$, as $Q_i$ and $q_i$, respectively. The subtree of $G$ described so far is denoted by $H$. Finally, for each $i \in \{1,\ldots,k\}$, there is a tree $T_{0,i}$ with root $p_i$, where $T_{0,i}$ is a path of length 2. The port sequence from each $p_i$ to the other endpoint of $T_{0,i}$ is $(1,0,1,0)$. Further, for each $i \in \{1,\ldots,k\}$, there is a tree $T_{1,i}$ of height 2 with $q_i$ as the root. More specifically, $T_{1,i}$ is a path of length 2 with an additional edge incident to the middle node. The port sequences from each $q_i$ to the leaves of $T_{1,i}$ are $(1,0,1,0)$ and $(1,0,2,0)$. This completes the definition of $G$. Figure \ref{lbdiagodd} illustrates tree $G$.

\begin{figure}[!ht]
\begin{center}
\includegraphics[scale=0.8]{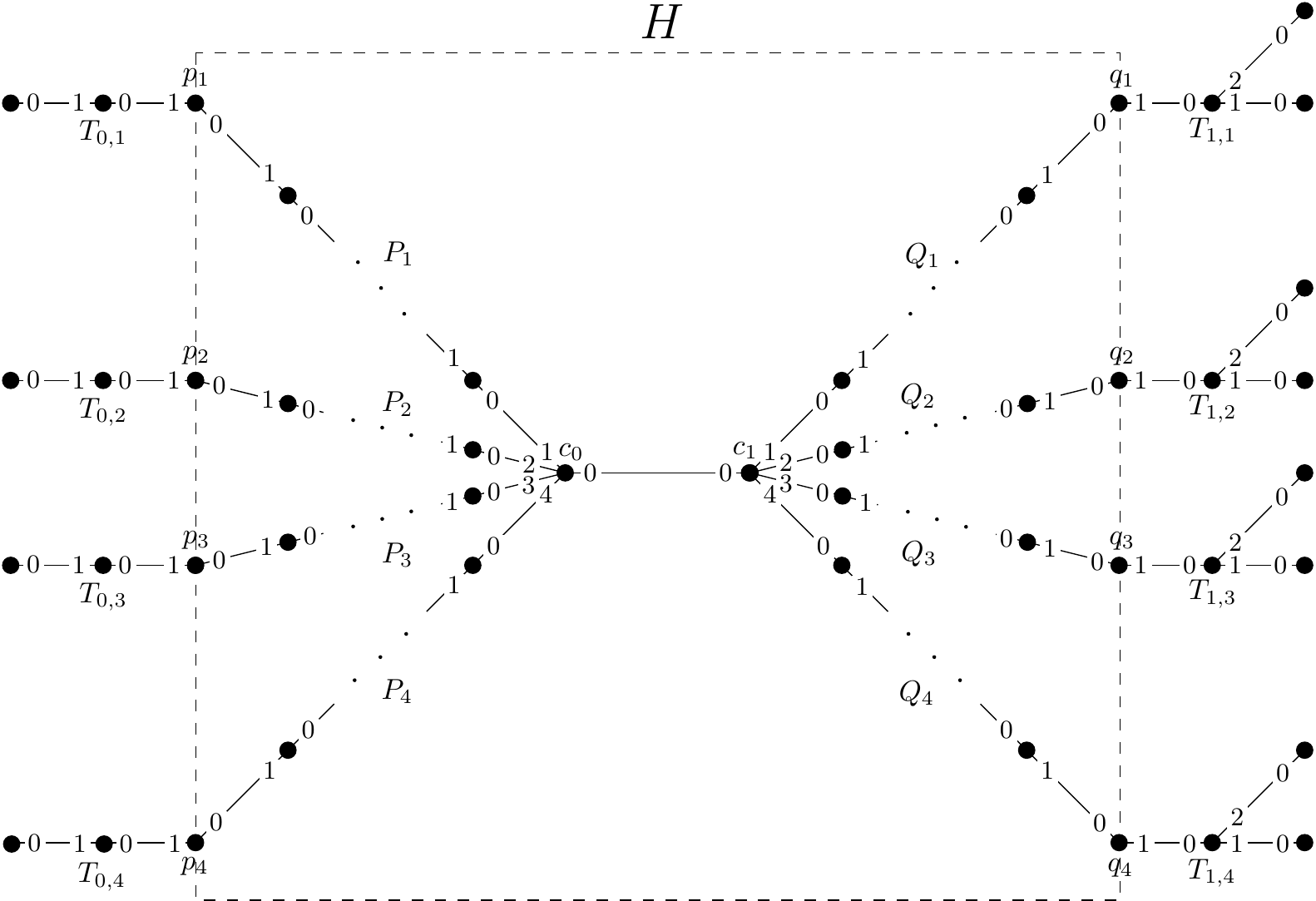}
\end{center}
\caption{Tree $G$ constructed in the proof of Theorem \ref{lb2odd}, with $k=4$.}
\label{lbdiagodd}
\end{figure}

Next, for every subset $\sigma$ of $\{2,\ldots,k\}$, we define a tree $G_{\sigma}$. At a high level, $G_{\sigma}$ is obtained from $G$ by swapping the subtrees rooted at $p_i$ and $q_i$, for each $i \in \sigma$. More specifically, the definition of $G_{\sigma}$ is similar to the definition of $G$ above, except that, for each $i \in \sigma$, tree $T_{0,i}$ has $q_i$ as its root and tree $T_{1,i}$ has $p_i$ as its root. See Figure \ref{swapped} for an example of $G_{\sigma}$. Note that $G_{\emptyset} = G$. Further, for any $\sigma \neq \sigma'$, we have $G_{\sigma} \neq G_{\sigma'}$. However, note that for any $\sigma$, since the differences between $G$ and $G_{\sigma}$ are only at the leaves or neighbours of leaves, we have that the subtree $H$ of $G$ is also a subtree of $G_{\sigma}$. The following result about $H$ follows from the symmetry of $H$ with respect to the central edge
and  from the fact that $p_i$ and $q_i$ are images of each other under this symmetry.

\begin{figure}[!ht]
\begin{center}
\includegraphics[scale=0.8]{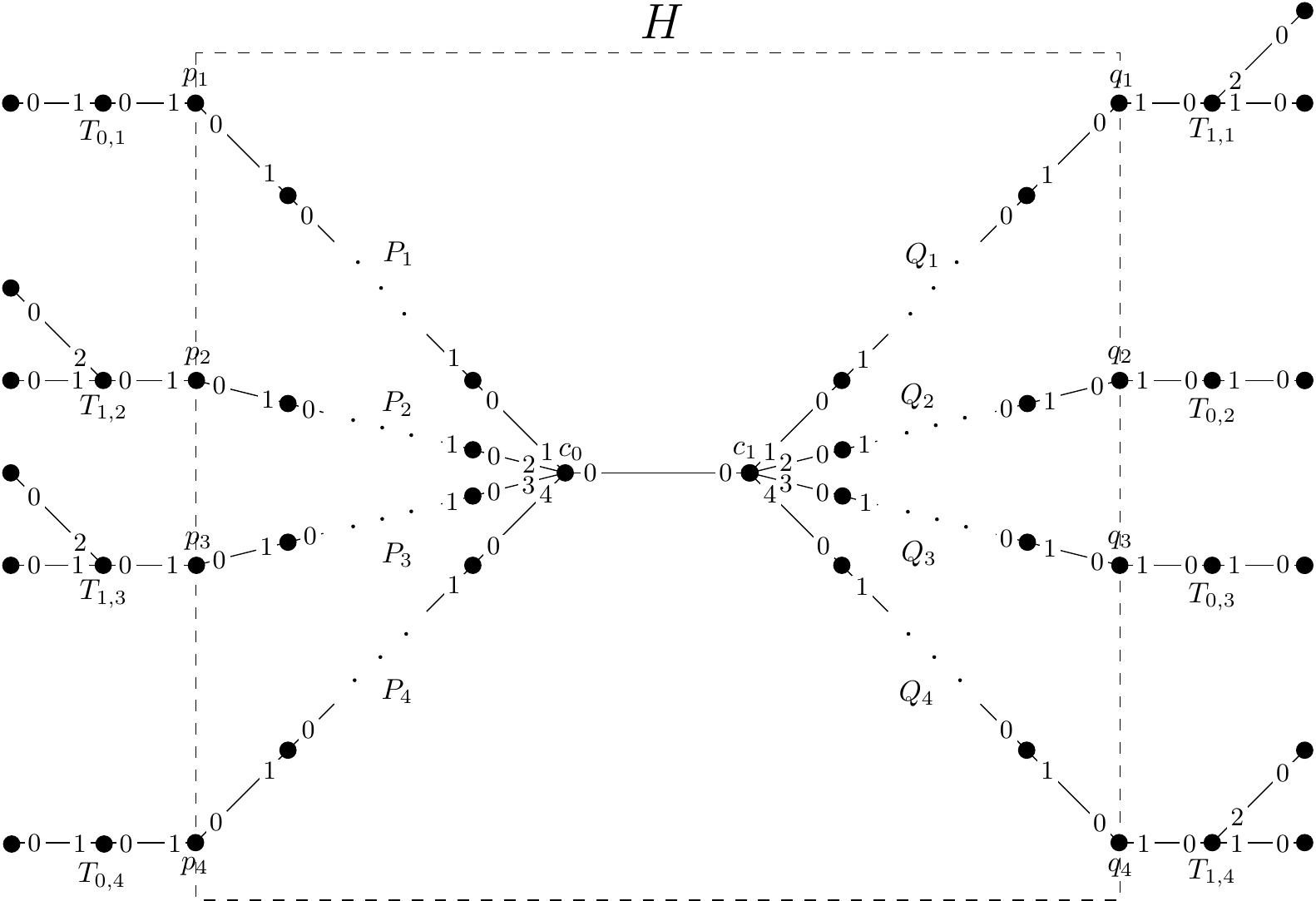}
\end{center}
\caption{Tree $G_{\sigma}$ constructed in the proof of Theorem \ref{lb2odd}, with $k=4$ and $\sigma = \{2,3\}$.}
\label{swapped}
\end{figure}

\begin{fact}\label{sameviewH}
For any $i \in \{2,\ldots,k\}$, the views $V_H(p_i,D-5)$ and $V_H(q_i,D-5)$ are identical.
\end{fact}

The class $\cal T$ is defined as the class of trees $G_{\sigma}$ for all subsets $\sigma$ of $\{2,\ldots,k\}$.
We now set out to prove a lower bound on the number of different advice strings needed to solve leader election for all trees in $\cal T$.

\begin{claim}
 For any $\sigma \neq \sigma'$, the advice strings provided to algorithm $ELECT$ for trees $G_{\sigma}$ and $G_{\sigma'}$ must be different.
.
\end{claim}

We prove the claim by contradiction. Assume that the advice strings assigned to $G_{\sigma}$ and $G_{\sigma'}$ are the same. Since $\sigma \neq \sigma'$, without loss of generality, assume that there is an integer $i \in \sigma$ such that $i \not\in \sigma'$. Consider a leaf $v$ of tree $T_{1,i}$, which, in $G_{\sigma}$, is rooted at node $p_i$. Note that,  $V_{G_{\sigma}}(v,D-3) = V_{T_{1,i}}(v,2) \cup V_H(p_i,D-5)$. In $G_{\sigma'}$, tree $T_{1,i}$ is rooted at node $q_i$, so $V_{G_{\sigma '}}(v,D-3) = V_{T_{1,i}}(v,2) \cup V_H(q_i,D-5)$. By Fact \ref{sameviewH}, it follows that $V_{G_{\sigma}}(v,D-3)=V_{G_{\sigma '}}(v,D-3)$. So, $v$ outputs the same sequence of outgoing ports after executing $ELECT$ in both $G_{\sigma}$ and $G_{\sigma'}$. Assume, without loss of generality, that in $G_{\sigma}$, the node elected by $v$ after executing $ELECT$ is closer to $c_0$ than to $c_1$. Since the length of $v$'s output is the same for both $G_{\sigma}$ and $G_{\sigma'}$, the node elected by $v$ after executing $ELECT$ in $G_{\sigma'}$ is closer to $c_1$ than to $c_0$. However, in both $G_{\sigma}$ and $G_{\sigma'}$, tree $T_{0,1}$ is rooted at node $p_i$, and the leaf $v'$ of tree $T_{0,1}$ outputs the same sequence of outgoing ports. Hence, $v'$ elects the same node in both $G_{\sigma}$ and $G_{\sigma'}$. Thus, in at least one of $G_{\sigma}$ or $G_{\sigma'}$, nodes $v$ and $v'$ do not elect the same node, a contradiction. This concludes the proof  of the claim.

Since there are $2^{k-1}$ different subsets of $\{2,\ldots,k\}$, the number of different advice strings is at least $2^{k-1}$. It follows that the size of advice is $\Omega(k) = \Omega(n/D)$.

Finally, for any $\sigma$, we prove that $\xi(G_{\sigma}) \leq h$, which implies that $\xi(G_{\sigma}) \leq \tau$. It is sufficient to show that, using $h$ rounds of communication, an arbitrary node $v$
given a map of $G_{\sigma}$  can compute where it is located in the map. The distance from any node $v$ to a node on the central edge is at most $h$. Let $c(v)$ be the endpoint of the central edge that is closest to $v$. It follows that $V(v,h)$ contains $c(v)$ (which can be identified by finding the closest node to $v$ that has degree $k+1$).
Consider two cases. If $v = c(v)$, then the subtree of $V(v,h)$ induced by the nodes that can be reached from $v$ via a path starting with port 1 can be used to uniquely identify whether $v=c_0$ or $v=c_1$. Indeed, $T_{0,1}$ is rooted at the node at distance $h-2$ from $v$ in this subtree if and only if $v=c_0$. If $v \neq c(v)$, let $i$ be the final port number in the port sequence $seq(v,c(v))$. The subtree of $V(c(v),h)$ induced by the nodes that can be reached from $c(v)$ via a path starting with port $i$ either has $T_{0,i}$ or $T_{1,i}$ rooted at the node at distance $h-2$ from $c(v)$. By identifying which of these two trees appears, $v$ can identify its position on a map of $G_{\sigma}$.
\end{proof}

\subsubsection{Even diameter}

The lower bound argument for even diameter closely resembles that for odd diameter, as given in the previous theorem. 
However, in this case, it holds only for $ \tau \leq D - 4$. At a high level, we construct trees of even diameter by increasing by 1 the diameter of trees constructed in the previous case. On the other hand, we decrease the time by 1, so that the views of certain nodes do not change. 
Rather than providing a list of small but numerous changes in the proof, 
we give the entire modified construction and argument for the reader's convenience.

\begin{theorem}\label{lb2even}
Let $8 \leq D<n$ be positive integers, where $D$ is even. 
Fix any value $D/2 \leq \tau \leq D - 4$.
There exists a class $\cal T$ of trees $T$ with size $\Theta(n)$, diameter $D$, and $\xi(T) \leq \tau$, such that
every leader election algorithm working in time $\tau$ on the class $\cal T$ requires advice of size $\Omega(n/D)$.
\end{theorem}
\begin{proof}
Let $h = \frac{D-2}{2}$ and let $k  = \left\lceil \frac{n}{D-1} \right\rceil$. Let $m = (D-1)k + 2 \in \Theta(n)$. We define a class of trees $T$ with  size $m \in \Theta(n)$,
even diameter $D$, and  $\xi(T) \leq \frac{D}{2}$ such that the minimum size of advice needed by an arbitrary algorithm $ELECT$ solving leader election in time $\tau$  for this class is $\Omega(n/D)$.

We start with a single tree $G$ of size $m$, defined as follows. Consider a single edge $\{c_0,c_1\}$, and let the port numbers corresponding to this edge at $c_0$ and $c_1$ be both 0. Add a path $P_1$ of length $h$ with $c_0$ as one endpoint, and denote the other endpoint of this path by $p_1$. Add a path $Q_1$ of length $h+1$ with $c_1$ as one endpoint, and denote the other endpoint of this path by $q_1$. Let the port sequences $seq(c_0,p_1)$ and $seq(c_1,q_1)$ be $(1,0,\ldots,1,0)$ (where the length of the latter sequence is two greater than the former.) 

Next, for each $i \in \{2,\ldots,k\}$, add a path $P_i$ of length $h-2$ with $c_0$ as one endpoint. The other endpoint of each of these paths will be denoted by $p_i$. Further, let the port sequence $seq(c_0,p_i)$ is equal to $(i,0,1,0,1,0,\ldots,1,0)$. Add the same paths with $c_1$ as one endpoint. For each $i \in \{2,\ldots,k\}$, we will refer to each of these paths, and their corresponding endpoint other than $c_1$, as $Q_i$ and $q_i$, respectively. The subtree of $G$ described so far is denoted by $H$. Finally, for each $i \in \{2,\ldots,k\}$, add a tree $T_{0,i}$ with root $p_i$, where $T_{0,i}$ is a path of length 2. Let the port sequence from each $p_i$ to the other endpoint of $T_{0,i}$ be $(1,0,1,0)$. Further, for each $i \in \{2,\ldots,k\}$, add a tree $T_{1,i}$ of height 2 with $q_i$ as the root. More specifically, $T_{1,i}$ is a path of length 2 with an additional edge incident to the middle node. Let the port sequences from each $q_i$ to the leaves of $T_{1,i}$ be $(1,0,1,0)$ and $(1,0,2,0)$. This completes the definition of $G$. Figure \ref{lbdiageven} illustrates tree $G$. 

\begin{figure}[!ht]
\begin{center}
\includegraphics[scale=0.8]{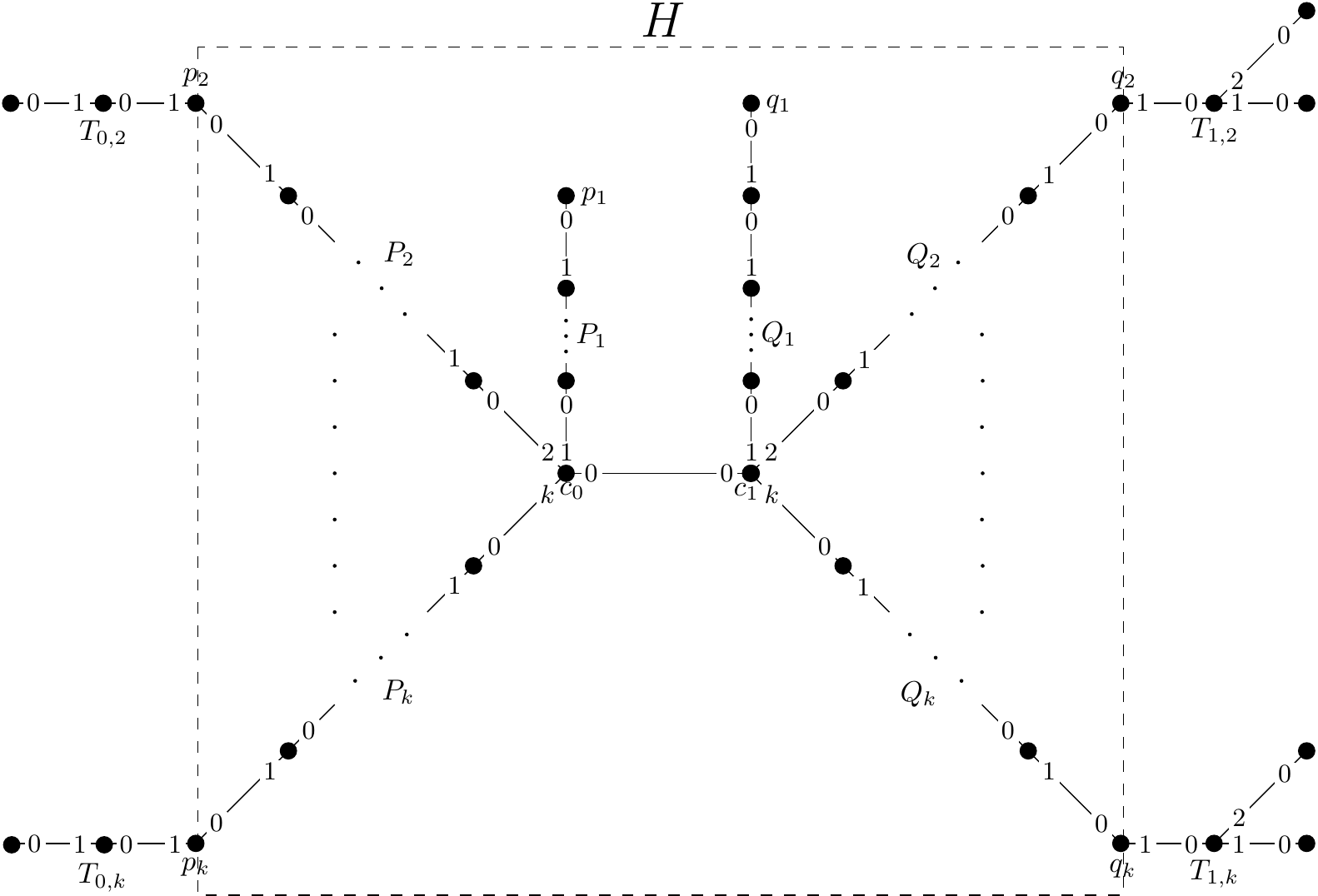}
\end{center}
\caption{Tree $G$ constructed in the proof of Theorem \ref{lb2even}}
\label{lbdiageven}
\end{figure}

Next, for every subset $\sigma$ of $\{2,\ldots,k\}$, we define a tree $G_{\sigma}$. At a high level, $G_{\sigma}$ is obtained from $G$ by swapping the subtrees rooted at $p_i$ and $q_i$, for each $i \in \sigma$. More specifically, the definition of $G_{\sigma}$ is similar to the definition of $G$ above, except that, for each $i \in \sigma$, tree $T_{0,i}$ has $q_i$ as its root and tree $T_{1,i}$ has $p_i$ as its root. Note that $G_{\emptyset} = G$. Further, for any $\sigma \neq \sigma'$, we have $G_{\sigma} \neq G_{\sigma'}$. However, note that for any $\sigma$, since the differences between $G$ and $G_{\sigma}$ are only at the leaves or neighbours of leaves, we have that the subtree $H$ of $G$ is also a subtree of $G_{\sigma}$. The following result about $H$ follows from: the symmetry of $H$, with $P_1$ and $Q_1$ removed, with respect to the central edge,
the fact that, for $i \in \{2,\ldots,k\}$, $p_i$ and $q_i$ are images of each other under this symmetry, and the fact that $p_i$ and $q_i$ cannot deduce the lengths
of $P_1$ and $Q_1$ in time $D-6$.

\begin{fact}\label{sameviewHeven}
For any $i \in \{2,\ldots,k\}$, the views $V_H(p_i,D-6)$ and $V_H(q_i,D-6)$ are identical.
\end{fact}

The class $\cal T$ is defined as the class of trees $G_{\sigma}$ for all subsets $\sigma$ of $\{2,\ldots,k\}$.
We now set out to prove a lower bound on the number of different advice strings needed to solve leader election for all trees in $\cal T$.

\begin{claim}
For any $\sigma \neq \sigma'$, the advice strings provided to algorithm $ELECT$ for trees $G_{\sigma}$ and $G_{\sigma'}$ must be different.
\end{claim}

We prove the claim by contradiction. Assume that the advice strings assigned to $G_{\sigma}$ and $G_{\sigma'}$ are the same. Since $\sigma \neq \sigma'$, without loss of generality, assume that there is an integer $i \in \sigma$ such that $i \not\in \sigma'$. Consider a leaf $v$ of tree $T_{1,i}$, which, in $G_{\sigma}$, is rooted at node $p_i$. Note that, in $G_{\sigma}$, $V_{G_{\sigma}}(v,D-4) = V_{T_{1,i}}(v,2) \cup V_H(p_i,D-6)$. In $G_{\sigma'}$, tree $T_{1,i}$ is rooted at node $q_i$, so $V_{G_{\sigma '}}(v,D-4) = V_{T_{1,i}}(v,2) \cup V_H(q_i,D-6)$. By Fact \ref{sameviewHeven}, it follows that $V_{G_{\sigma}}(v,D-4)=V_{G_{\sigma '}}(v,D-4)$. 
So, $v$ outputs the same sequence of outgoing ports after executing $ELECT$ in both $G_{\sigma}$ and $G_{\sigma'}$. Assume, without loss of generality, that in $G_{\sigma}$, the node elected by $v$ after executing $ELECT$ is closer to $c_0$ than to $c_1$. Since the length of $v$'s output is the same for both $G_{\sigma}$ and $G_{\sigma'}$, the node elected by $v$ after executing $ELECT$ in $G_{\sigma'}$ is closer to $c_1$ than to $c_0$. However, in both $G_{\sigma}$ and $G_{\sigma'}$, tree $T_{0,1}$ is rooted at node $p_i$, and the leaf $v'$ of tree $T_{0,1}$ outputs the same sequence of outgoing ports. Hence, $v'$ elects the same node in both $G_{\sigma}$ and $G_{\sigma'}$. Thus, in at least one of $G_{\sigma}$ or $G_{\sigma'}$, nodes $v$ and $v'$ do not elect the same node, a contradiction. This concludes the proof  of the claim.

Since there are $2^{k-1}$ different subsets of $\{2,\ldots,k\}$, the number of different advice strings is at least $2^{k-1}$. It follows that the size of advice is $\Omega(k) = \Omega(n/D)$.

Finally, for any $\sigma$, we prove that $\xi(G_{\sigma}) \leq h+1$, which implies that $\xi(G_{\sigma}) \leq \tau$. It is sufficient to show that, using $h+1$ rounds of communication, an arbitrary node $v$
given a map of $G_{\sigma}$  can compute where it is located in the map. The distance from any node $v$ to a node on the edge $\{c_0,c_1\}$ is at most $h+1$. Let $c(v)$ be the endpoint of the edge $\{c_0,c_1\}$ that is closest to $v$. It follows that $V(v,h+1)$ contains $c(v)$ (which can be identified by finding the closest node to $v$ that has degree $k+1$).
Consider two cases. If $v = c(v)$ or $v$ is in the the subtree $T_1$ of $V(c(v),h+1)$ induced by the nodes that can be reached from $v$ via a path starting with port 1, then $T_1 \in \{P_1,Q_1\}$ can be used to identify $v$'s position on a map of $G_{\sigma}$. Indeed, $v$ need only check the length of $T_1$, since $|T_1|=h$ if and only if $T_1 = P_1$. In the second case, let $i$ be the final port number in the port sequence $seq(v,c(v))$. The subtree of $V(c(v),h+1)$ induced by the nodes that can be reached from $c(v)$ via a path starting with port $1$ has $T_{0,i}$ or $T_{1,i}$ rooted at the node at distance $h-2$ from $c(v)$. By identifying which of these two trees appears, $v$ can identify its position on a map of $G_{\sigma}$.
\end{proof}

\section{Time $\alpha \cdot diam$ for constant $\alpha <1/2$}\label{alphadiam}

In this section, we prove tight upper and lower bounds of $\Theta(n)$ on the minimum size of advice sufficient to perform leader election
 in time $\alpha D$ for constant $\alpha <1/2$ in $n$-node trees $T$ with diameter $D \in \omega (\log ^2 n)$ and $\xi(T) \leq \alpha D$.
The upper bound, which holds for all values of $D$, is given by the following result. 

\begin{proposition}\label{O(n)}
Leader election in every non symmetric $n$-node  tree $T$ is possible in time $\xi(T)$, using $O(n)$ bits of advice.
\end{proposition}

\begin{proof}
We use the following observation of Chierichetti \cite{Ch}.
An $n$-node anonymous tree can be coded by an ordered pair of two sequences $(\phi, \psi)$  in such a way
that trees  that are not port-preserving isomorphic get different codes. 
Starting from any node considered as a root, perform a DFS traversal of the tree, visiting children of any node in the
increasing order of ports at this node. 
The binary sequence $\phi$ has length $2(n-1)$ and is  
constructed as follows.
Whenever an edge is traversed down the tree, write 0, and whenever it is traversed up the tree, write 1. 
The sequence $\psi$ has length $n-1$: let $(v_1,\dots ,v_{n-1})$ be the sequence of nodes other than the root in the order of first visit in the traversal, and let the $i^{th}$ term
of $\psi$ be the entry port number at the first visit of $v_i$. 
There are at most  $2^{2(n-1)}$ possible sequences $\phi$. The number of possible sequences $\psi$ for each $\phi$ is bounded above by the product of degrees of all nodes other
than the root and the sum of these degrees is at most $2(n-1)$. Hence the number of  possible sequences $\psi$ for each $\phi$ is bounded above by $2^{n-1}$. Hence,
there are at most $2^{2(n-1)} \cdot 2^{n-1}$ possible codes. The code of a tree is the lexicographically smallest pair $(\phi, \psi)$ over all choices of the root.

We solve leader election as follows.
The advice is the code of the tree. It has size $O(n)$. Using this code, all nodes  construct a faithful map of the tree. Using the map, they can perform leader election in time
$\xi(T)$, by the definition of this parameter. 
\end{proof}

The next result is a matching lower bound when the diameter is not too small compared to $n$.

\begin{theorem}\label{Omega(n)}
Let $D<n$ be positive integers, such that $D \in \omega(\log^2n)$. Let $\alpha < 1/2$ be a constant and let $\tau= \lfloor \alpha D \rfloor$. There exist trees $T$ with  size $\Theta(n)$, diameter $D$, and $\xi(T) \leq \tau$, for which any leader election algorithm working in time $\tau$
requires advice of size $\Omega(n)$.
\end{theorem}
\begin{proof}
We define a class ${\cal T}$ of trees $T$ with  size $\Theta(n)$, diameter $D$, and $\xi(T) \leq \tau$ such that the minimum size of advice needed by an arbitrary algorithm $ELECT$ solving leader election in time $\tau$  for this class is $\Omega(n)$.

We consider the case where $D$ is even (the case where $D$ is odd is obtained by adding a single edge to the construction for diameter $D-1$). We begin by constructing a ``template'' tree $G$ from which all trees in our class ${\cal T}$ will be constructed. Tree $G$ itself is not a valid instance for leader election since some of its port numbers are undefined, but each tree in our class is obtained from $G$ by filling in the missing port numbers. Our construction of $G$ depends on a set ${\cal M}$ of trees that we call \emph{markers}. Each marker is a tree of height 2 with fixed port numbers, and each marker appears in $G$ only once. The purpose of the markers is to guarantee $\xi(T) \leq \tau$ by enabling each node to determine its location in a map of $G$. Later, we will specify how to define the markers so that we have as many of them as we need.

Let $k = \lceil 2n/D \rceil$. Our template $G$ consists of a central node $c$, and, for each $i \in \{0,\ldots,k-1\}$, a path $P_i$ of length $D/2$ with one endpoint equal to $c$. For each $i \in \{0,\ldots,k-1\}$, the other endpoint of $P_i$  will be denoted by $p_i$, and the first port number of the sequence $seq(c,p_i)$ is equal to $i$. Let $f$ be the integer in $\{\tau +1,\tau+2\}$ that has the same parity as $D/2$. The integer $f$ is the number of nodes on each path $P_i$ (without $c$) whose incident port numbers are fixed in $G$. More specifically, the first $2f-1$ port numbers of the sequence $seq(p_i,c)$ are $(0,1,0,1,\ldots,0)$. Finally, for each $i \in \{0,\ldots,k-1\}$, we place one marker rooted at distance 2 from $p_i$, and, from this node, we place one marker rooted at every $(\tau - 4)^{th}$ node on the path towards $c$. We ensure that no two markers from ${\cal M}$ are used twice in the construction. See Figure \ref{lbOmegan} for an illustration of $G$.

\begin{figure}[!ht]
\begin{center}
\includegraphics[scale=0.8]{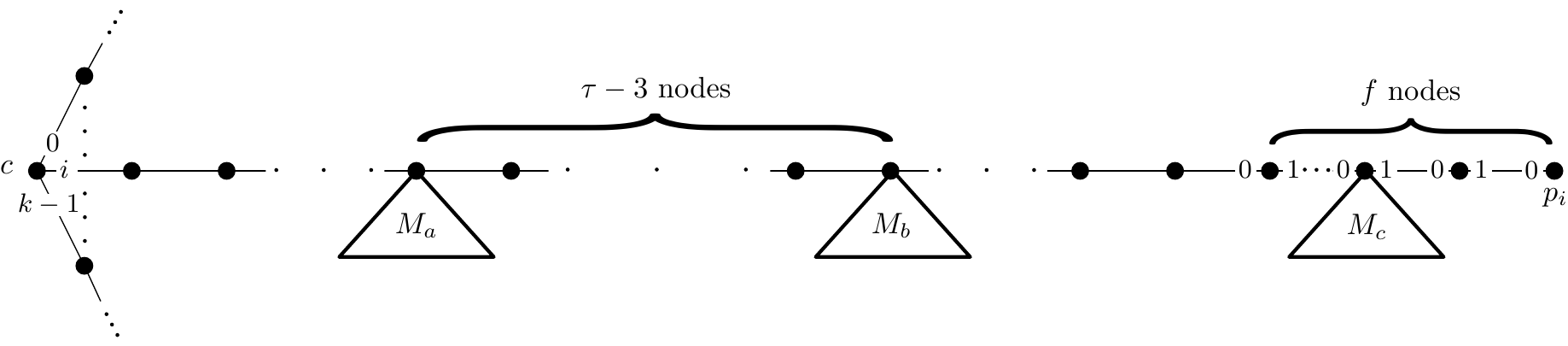}
\end{center}
\caption{Template $G$ constructed in the proof of Theorem \ref{Omega(n)}, where $M_a$, $M_b$, $M_c$ are different markers from $\cal M$.}
\label{lbOmegan}
\end{figure}

We are now ready to construct the class ${\cal T}$. Let $s = \frac{D}{2} - f$. Hence $s \in \Theta(D)$. We note that, on each path $P_i$ in $G$, there are $s$ consecutive nodes (starting at $c$'s neighbour on $P_i$) whose incident port numbers are not defined. Call these nodes $v_{i,1},\ldots,v_{i,s}$. Since $f$ was chosen to have the same parity as $\frac{D}{2}$, it follows that $s$ is even. For $j \in \{1,\ldots,s/2\}$, let $e_{i,j}$ be the edge $\{v_{i,2j-1},v_{i,2j}\}$. We say that edge $e_{i,j}$ is \emph{set to 0} (respectively, \emph{set to 1}) if the two ports corresponding to edge $e_{i,j}$ are equal to 0 (respectively, equal to 1), and the ports at $v_{i,2j-1}$ and $v_{i,2j}$ not corresponding to edge $e_{i,j}$ are equal to 1 (respectively, equal to 0). We now demonstrate how to obtain a tree $G_{L} $ by defining a labeling function $L$. In particular, a \emph{labeling function} $L : \{0,\ldots,k-1\} \times \{1,\ldots,s/2\} \longrightarrow \{0,1\}$ maps pairs of integers $(i,j)$ to 0 or 1. The tree $G_{L}$ is obtained from $G$ by setting each edge $e_{i,j}$ to $L(i,j)$. Note that, for any labeling function $L$, the port labeling of tree $G_L$ is valid. Further, for any distinct labeling functions $L,L'$, the trees $G_{L}$ and $G_{L'}$ are distinct. The class $\cal T$ is defined as the class of trees $G_L$ for all labeling functions $L$. The following result will help us compute the number of different advice strings required by algorithm $ELECT$ for the class ${\cal T}$.

\begin{claim}\label{Ldifferent}
 For any distinct labeling functions $L,L'$, 
 the advice strings provided to algorithm $ELECT$ for trees $G_{L}$ and $G_{L'}$ must be different.
\end{claim}
We prove the claim by contradiction. Assume that the advice strings assigned to $G_{L}$ and $G_{L'}$ are the same. Since $L$ and $L'$ are distinct, there exists an $i \in \{0,\ldots,k-1\}$ and a $j \in \{1,\ldots,s/2\}$ such that $L(i,j) \neq L'(i,j)$. Without loss of generality, assume that $L(i,j) = 0$. Choose an arbitrary $i' \in \{0,\ldots,k-1\}-\{i\}$. We consider the executions of algorithm $ELECT$ by nodes $p_{i}$ and $p_{i'}$ in both $G_{L}$ and $G_{L'}$. First, note that the first $2f-1 \geq 2\tau + 1$ port numbers of $seq(p_{i},c)$ and $seq(p_{i'},c)$ were fixed in $G$. It follows that $V(p_i,\tau)$ is the same in both $G_{L}$ and $G_{L'}$. Similarly, $V(p_{i'}, \tau)$ is the same in both $G_{L}$ and $G_{L'}$. Since $G_L$ and $G_{L'}$ are assigned the same advice, it follows that $p_i$ outputs the same sequence of outgoing ports, say $\sigma$, after executing $ELECT$ in $G_L$ as it does after executing $ELECT$ in $G_{L'}$. Similarly, $p_{i'}$ outputs the same sequence of outgoing ports, say $\sigma'$, in both executions. Since the sequence outputted by $p_i$ is the same in both executions, corresponds to simple paths in both of them, and the port numbers at $c$
are fixed, $p_i$ elects the same node in $G_L$ and in $G_L'$. The same is true for  $p_{i'}$.
Therefore, it must be the case that, for one of $\sigma$ or $\sigma'$, the corresponding path crosses edge $e_{i,j}$ using the same outgoing port number in the execution in $G_L$ as in the execution in $G_{L'}$. However, the two ports corresponding to $e_{i,j}$ are labeled 0 in $G_L$ and labeled 1 in $G_{L'}$, a contradiction. This proves the claim.

By Claim \ref{Ldifferent}, the number of different advice strings is equal to the number of distinct labeling functions, i.e., $2^{k(s/2)}$. But $s = \frac{D}{2} - f
\geq \frac{D}{2} -(\tau +2)  \geq \frac{D}{2} - \lfloor \alpha D \rfloor - 2 \in \Theta(D)$, and $k=\lceil 2n/D \rceil$, so the number of labeling functions is $ 2^{cn}$ for some positive constant $c$. It follows that the minimum number of bits sufficient to encode the advice strings is $\Omega(n)$, as required.

Next, we show that, for each $T \in {\cal T}$, we have $\xi(T) \leq \tau$. By construction, for each node $v$ in $T$, at least one of the following must be true:
\begin{enumerate}
\item $v$ is contained in a marker \label{inmarker}
\item $v$ is located between two consecutive roots of markers \label{betweentwo}
\item for some $i \in \{0,\ldots,k-1\}$, $v=p_i$, or $v$ is located between the root of a marker and $p_i$, with no root of a marker between $v$ and $p_i$\label{atleaf}
\item $v=c$, or, $v$ is located between the root of a marker and $c$, with no root of a marker between $v$ and $c$ \label{nearcentre}
\end{enumerate}

In case (1), $v$ is at distance at most 2 from the root of the marker. Since the markers are spaced distance $(\tau-4)$ apart, the distance from $v$ to a neighbouring marker is at most $\tau-2$. Therefore, $V(v,\tau)$ contains the marker containing $v$ as well as at least one neighbouring marker. In case (2), since the markers are spaced distance $(\tau-4)$ apart and the height of each marker is 2, it follows that $V(v,\tau)$ contains the two markers closest to $v$. In both cases, using these two markers, $v$ can determine where on a map of $T$ it is located. In case (3), since we placed a marker rooted at distance at most 2 from each $p_i$, 
it follows that $V(v,\tau)$ contains the marker closest to $v$. Using this marker, $v$ can determine where on a map of $T$ it is located. In case (4), 
by the construction of $G$, the distance between $v$ and $c$ is at most $\tau -4$, so $c$ is in $V(v,\tau)$. Node $c$ is uniquely recognizable as the node of degree $k$, as long as we ensure that no root of a marker has degree $k$ (see the specification of the markers below). 

Hence in all cases $v$ can locate itself in the map within time $\tau$, so leader election can be done in this time given the map, thus proving $\xi(T) \leq \tau$.

It remains to describe the set of markers. First, observe that the total number of markers needed to define the template $G$ is bounded above by $k\left\lceil\frac{D}{2(\tau-4)}\right\rceil$. This is because, for each $i \in \{0,\ldots,k-1\}$, path $P_i$ in $G$ contains $\frac{D}{2}$ nodes (other than $c$), and a marker is placed every $(\tau-4)$ nodes. So the number of markers needed is bounded above by $4\left(\frac{2n}{D}\right)\left(\frac{D}{2(\tau-4)}\right) = \frac{4n}{\tau-4}$.

Let $y = \lceil \log\frac{4n}{\tau-4} \rceil$. Consider the family $\cal X$ of all trees of height 2 with $z=y^2$ leaves and whose root has degree $x = \lceil y^{3/2} \rceil\leq y^{7/4}$,
for sufficiently large $n$. For each of these trees, label the ports at the root node using $\{2,\ldots,x+1\}$ (the port numbers 0 and 1 are reserved to label the ports on path $P_i$.) For each node at the first level, the port leading towards the root node is labeled 0. The number of trees in $\cal X$ is equal to the number of ordered partitions of the set of leaves into $x$ parts (a partition specifies the number of leaves adjacent to each of the $x$ nodes at the first level of the marker.) So the number of trees is equal to $\binom{z+x}{x} = \frac{(z+x)\cdots(z+1)}{x!} \geq \frac{z^x}{x!}$. We now show that $\frac{z^x}{x!} \geq \frac{n}{D}$. Note that $\log\left( \frac{z^x}{x!} \right) = \log{z^x} - \log{x!} \geq  x\log{z} - x\log{x}\geq  \lceil y^{3/2} \rceil\log{y^2} -  \lceil y^{3/2} \rceil\log{y^{7/4}}= 2 \lceil y^{3/2} \rceil\log{y} - \frac{7}{4}  \lceil y^{3/2} \rceil\log{y} \geq y$. It follows that $\frac{z^x}{x!} \geq 2^y \geq \frac{4n}{\tau-4}$. Therefore, by taking an arbitrary subset $\cal M$ of size at least $\frac{4n}{\tau-4}$ of the family $\cal X$, we have a sufficiently large set of markers to define $G$. Note that, since the size of each marker is $x+z$, the total number of nodes needed to define the markers in $G$ is at most $\lfloor \frac{4n}{\tau-4}\rfloor (y^{7/4} + y^2) \leq \frac{8n}{\tau-4}(y^2) \leq \frac{32n}{\tau-4}\log^2\frac{4n}{\tau-4}$. Since $D \in \omega(\log^2{n})$, it follows that $\tau - 4 = \lfloor \alpha D \rfloor - 4 \in \omega(\log^2\frac{4n}{\tau-4})$. Hence, for sufficiently large $n$, the number of nodes needed to define the markers is $o(n)$. Thus, the size of $G$ (and, therefore, of each tree in ${\cal T}$) is in $\Theta(n)$.
\end{proof}

Proposition \ref{O(n)} and Theorem \ref{Omega(n)} imply the following corollary.

\begin{corollary}
Let $D<n$ be positive integers, such that $D \in \omega(\log^2n)$. Let $\alpha < 1/2$ be a constant and let $\tau= \lfloor \alpha D \rfloor$.
Let $\cal T$ be the class of trees $T$  with size $\Theta(n)$, diameter $D$, and $\xi(T) \leq \tau$.
The minimum size of advice to perform leader election in time $\tau$ for the class $\cal T$ is $\Theta(n)$.
\end{corollary}

\section{Discussion of open problems}

For time values $diam -1$ and $diam-2$, we gave tight bounds (up to constant factors) on the minimum size of advice sufficient to perform leader election for trees of diameter $diam$. For time in the interval  $[\beta \cdot diam, diam -3]$ for constant $\beta >1/2$, we gave bounds leaving a logarithmic gap, except for the special case
when $diam$ is even and time is exactly $diam-3$.
This yields the first problem: 

{\bf P1.} Find close upper and lower bounds on the minimum size of advice in the special case when $diam$ is even and time is exactly $diam-3$.

As a step in this direction, we prove the following lower bound which implies an exponential jump in the minimum size of advice
between time $diam -2$ and time $diam -3$ when $diam$ is even and constant.

\begin{proposition}\label{discussprop}
Let $n$ be a positive integer and let $D$ be an even positive integer constant such that $ D\geq 6$. There exists a class $\cal T$ of trees $T$ with size $\Theta(n)$, diameter $D$, and $\xi(T) \leq D-3$, such that
every leader election algorithm working in time $D-3$ on the class $\cal T$ requires advice of size $\Omega(n^{2/D}\log n)$.
\end{proposition}

\begin{proof}
Let $h = D/2$ and let $\Delta = \lfloor (\gamma n)^{1/h} \rfloor$, where $\gamma$ is a sufficiently large constant greater than 1 so that $\Delta \geq 2$. We start by defining a tree $G$ from which we will derive the class $\mathcal{T}$.
At a high level, trees in $\mathcal{T}$ will be defined in such a way that, for some leaves $v$, there are many nodes
in the view $V(v,D-3)$ such that $v$ cannot be sure which of them is the central node.   

To this end, in the construction of $G$, we will use special subtrees, called \emph{confusion subtrees}, as building blocks. For any $i \in \{0,\ldots,\Delta-1\}$, a {\em confusion subtree avoiding port $i$}, denoted by $T_i(x)$, is defined recursively with its height $x$ as parameter. For the base cases, let $T_i(0)$ consist of a single node with degree 0, and let $T_i(1)$ consist of a node $w_i$ with $\Delta$ degree-1 neighbours, such that the ports at $w_i$ are labeled by integers from $ \{0,\ldots,\Delta\} \setminus \{i\}$. For any $x>1$, define $T_i(x)$ as follows:
\begin{enumerate}
\item let $w_i$ be the root node with degree $\Delta$
\item Label one of $w_i$'s neighbours as $c_{\Delta}$. Attach a path of length $x-2$ with $c_{\Delta}$ as one endpoint, and let $p_{\Delta}$ be the other endpoint of this path (if $x=2$, set $p_{\Delta} = c_{\Delta}$.) Set the port sequence $seq(w_i,p_{\Delta})$ equal to $(\Delta,0,1,0,1,\ldots,0,1,0)$. \label{TiAddPath}
\item for each $j \in \{0,\ldots,\Delta-1\} \setminus \{i\}$, label one neighbour of $w_i$ as $c_j$. Set the two ports corresponding to edge $\{w_i,c_j\}$ equal to $j$. For each $k \in \{0,\ldots,\Delta-1\} \setminus \{j\}$, add a neighbour $b_{j,k}$ to $c_j$. Set the two ports corresponding to edge $\{c_j,b_{j,k}\}$ equal to $k$. At each $b_{j,k}$, attach a copy of $T_k(x-2)$ by identifying the root of $T_k(x-2)$ with $b_{j,k}$.
\end{enumerate}

See Figure \ref{confusion} for an illustration of $T_i(x)$. Finally, let $G$ consist of a root node $c$ of degree $\Delta$, with the roots of $T_0(h-1),\ldots,T_{\Delta-1}(h-1)$ as its neighbours. For each $i \in \{0,\ldots,\Delta-1\}$, we denote by $w_i$ the root of $T_i(h-1)$, and we denote by $q_i$ the node $p_{\Delta}$ defined at step \ref{TiAddPath} in the construction of $T_i(h-1)$. For each $i \in \{0,\ldots,\Delta-1\}$, the two ports corresponding to the edge $\{c,w_i\}$ are set to $i$, and, $\Delta +i+1$ leaves are added as neighbours of $q_i$. See Figure \ref{confusion} for an illustration of $G$. 

\begin{figure}[!ht]
\begin{center}
\includegraphics[scale=0.85]{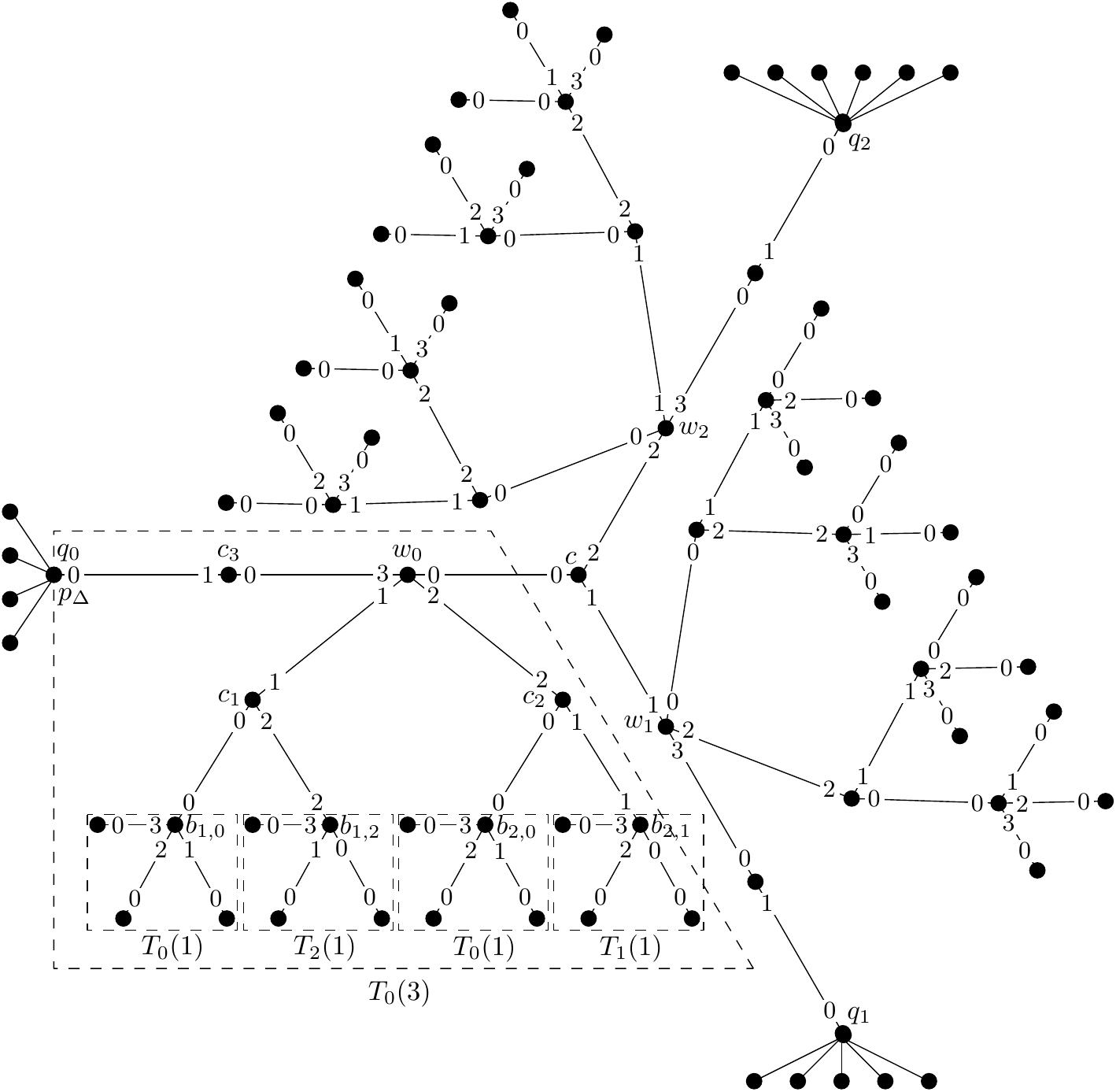}
\end{center}
\caption{Tree $G$ constructed in Proposition \ref{discussprop}, with $\Delta=3$ and $h=4$. Confusion trees $T_0(3), T_0(1), T_1(1), T_2(1)$ can be found in the dashed boxes.}
\label{confusion}
\end{figure}

We can now define the class $\mathcal{T}$. For every permuation $\sigma$ of the integers $\{1,\ldots,\Delta-1\}$, we obtain a tree $G_\sigma$ from $G$ by applying $\sigma$ to the subtrees rooted at $q_1,\ldots,q_{\Delta-1}$. More specifically, for each $i \in \{1,\ldots,\Delta-1\}$, the number of neighbouring leaves of $q_i$ is changed to $\Delta + \sigma(i) + 1$. It follows that, when $\sigma$ is the identity permutation, $G_{\sigma} = G$. Further, note that the only differences between two trees $G_{\sigma}$ and $G_{\sigma'}$ are the degrees of some nodes among $q_1,\ldots,q_{\Delta-1}$. The class $\mathcal{T}$ is defined as the set of trees $G_\sigma$ for all permutations $\sigma$ of the integers $\{1,\ldots,\Delta-1\}$. By counting the number of such permutations, it follows that $|\mathcal{T}| = (\Delta-1)!$. 

We now determine the number of different advice strings needed by any algorithm $ELECT$ that solves leader election in time $D-3$ in trees from the class $\mathcal{T}$. The following result shows that the number of different advice strings needed by algorithm $ELECT$ is $|\mathcal{T}|$.

\begin{claim}\label{different}
For any two distinct permuations $\sigma, \sigma'$ of the integers $\{1,\ldots,\Delta-1\}$, the advice strings provided to algorithm $ELECT$ for trees $G_{\sigma}$ and $G_{\sigma'}$ must be different.
\end{claim}
We prove the claim by contradiction. Assume that, for some permutations $\sigma \neq \sigma'$, the advice strings assigned to $G_{\sigma}$ and $G_{\sigma'}$ are the same. The proof proceeds by finding a leaf $\ell_{\sigma}$ in $G_{\sigma}$ and a leaf $\ell_{\sigma'}$ in $G_{\sigma'}$ such that the two leaves have the same view at distance $D-3$ in their respective trees. Since the two trees are assigned the same advice, this implies that the executions of $ELECT$ by these two leaves output the same port sequence to elect a leader in their respective trees. Further, we show that $q_0$ outputs the same port sequence in the execution of $ELECT$ in both $G_{\sigma}$ and $G_{\sigma'}$. Finally, we show that $ELECT$ fails in at least one of $G_{\sigma}$ or $G_{\sigma'}$, which gives the desired contradiction.

{\bf Definition of $\ell_{\sigma}$ and $\ell_{\sigma'}$.} Let $j \in \{1,\ldots,\Delta-1\}$ such that $\sigma(j) \neq \sigma'(j)$. Let $j' \in \{1,\ldots,\Delta-1\}$ such that $\sigma'(j') = \sigma(j)$. Let $\ell_{\sigma}$ be the leaf adjacent to $q_j$ in $G_{\sigma}$ such that the port at $q_j$ corresponding to the edge $\{q_j,\ell_{\sigma}\}$ is 1. Similarly, let $\ell_{\sigma'}$ be the leaf adjacent to $q_{j'}$ in $G_{\sigma'}$ such that the port at $q_{j'}$ corresponding to the edge $\{q_{j'},\ell_{\sigma'}\}$ is 1.

{\bf Showing that $\ell_{\sigma}$ and $\ell_{\sigma'}$ output the same port sequence.} Since $G_{\sigma}$ and $G_{\sigma'}$ are assigned the same advice, it is sufficient to show that $V_{G_{\sigma}}(\ell_{\sigma},D-3) = V_{G_{\sigma'}}(\ell_{\sigma'},D-3)$. To prove this fact, note that, since $\sigma(j) = \sigma'(j')$, it follows that $\ell_{\sigma}$'s neighbour $q_j$ in $G_{\sigma}$ has the same degree as $\ell_{\sigma'}$'s neighbour $q_{j'}$ in $G_{\sigma'}$ (i.e., both of these degrees are equal to $\Delta + \sigma(j) + 2$.) It follows that $V_{G_{\sigma}}(\ell_{\sigma},1) = V_{G_{\sigma'}}(\ell_{\sigma'},1)$. Next, by the constructions of $T_j(h-1)$ and $T_{j'}(h-1)$, we have that $seq(q_j,w_j) = (0,1,0,1,\ldots,0,\Delta)$ in $G_{\sigma}$ and that $seq(q_{j'},w_{j'}) = (0,1,0,1,\ldots,0,\Delta)$ in $G_{\sigma'}$. Since $|path(q_j,w_j)| = h-2 = |path(q_{j'},w_{j'})|$, we have so far shown that $V_{G_{\sigma}}(\ell_{\sigma},h-1) = V_{G_{\sigma'}}(\ell_{\sigma'},h-1)$. Next, the confusion trees were constructed specifically to satisfy the following property: for any port sequence of length $2h-3$ starting at $w_j$ in $G_{\sigma}$ that does not begin with port $\Delta$, the same sequence appears in $G_{\sigma'}$ starting at $w_{j'}$. It follows that $V_{G_{\sigma}}(w_j,h-2) = V_{G_{\sigma'}}(w_{j'},h-2)$. Finally, since $w_{j}$ is the only node at distance $h-1$ from $\ell_{\sigma}$ in $G_{\sigma}$, we have $V_{G_{\sigma}}(\ell_{\sigma},h-1) \cup V_{G_{\sigma}}(w_j,h-2) = V_{G_{\sigma}}(\ell_{\sigma},D-3)$. Similarly, since $w_{j'}$ is the only node at distance $h-1$ from $\ell_{\sigma'}$ in $G_{\sigma'}$, we have $V_{G_{\sigma'}}(\ell_{\sigma'},h-1) \cup V_{G_{\sigma'}}(w_{j'},h-2) = V_{G_{\sigma'}}(\ell_{\sigma'},D-3)$. Therefore, we have shown that $V_{G_{\sigma}}(\ell_{\sigma},D-3) = V_{G_{\sigma'}}(\ell_{\sigma'},D-3)$, as desired.

{\bf Showing that $q_0$ outputs the same port sequence in the execution of $ELECT$ in both $G_{\sigma}$ and $G_{\sigma'}$.} Since $G_{\sigma}$ and $G_{\sigma'}$ are assigned the same advice, it is sufficient to show that $V_{G_{\sigma}}(q_0,D-3) = V_{G_{\sigma'}}(q_0,D-3)$. To prove this fact, note that, since $d(q_j,c) = h-1$ for all $j \in \{0,\ldots,\Delta-1\}$, it follows that, for all $i \in \{1,\ldots,\Delta-1\}$, $d(q_0,q_i) = 2h-2 = D-2$. In particular, this means that $q_1,\ldots,q_{\Delta-1}$ are neither contained in $V_{G_{\sigma}}(q_0,D-3)$ nor contained in $V_{G_{\sigma'}}(q_0,D-3)$. Since the only differences between $G_{\sigma}$ and $G_{\sigma'}$ are the degrees of some nodes among $q_1,\ldots,q_{\Delta-1}$, it follows that $V_{G_{\sigma}}(q_0,D-3)$ and $V_{G_{\sigma'}}(q_0,D-3)$ must be equal.

{\bf Showing that $ELECT$ fails in at least one of $G_{\sigma}$ or $G_{\sigma'}$.} To obtain a contradiction, we assume that $ELECT$ correctly elects a leader in both $G_{\sigma}$ and $G_{\sigma'}$. First, suppose that $\ell_{\sigma}$ elects a node in $T_{j}(h-1)$ in $G_{\sigma}$, and $\ell_{\sigma'}$ elects a node in $T_{j'}(h-1)$ in $G_{\sigma'}$. It follows that, in $G_{\sigma}$, node $q_0$ elects a node in $T_j(h-1)$ by outputting some sequence $u$ whose $h^{th}$ term is $j$. Similarly, in $G_{\sigma'}$, node $q_0$ elects a node in $T_{j'}(h-1)$ by outputting some sequence $u'$ whose $h^{th}$ term is $j'\neq j$. However, this means that $u \neq u'$, which contradicts the fact that $q_0$ outputs the same port sequence in the execution of $ELECT$ in both $G_{\sigma}$ and $G_{\sigma'}$. So, we have shown that it is not the case that both $\ell_{\sigma}$ elects a node in $T_{j}(h-1)$ in $G_{\sigma}$ and $\ell_{\sigma'}$ elects a node in $T_{j'}(h-1)$ in $G_{\sigma'}$. So, without loss of generality, we may assume that, in $G_{\sigma}$, node $\ell_{\sigma}$ elects a node that is not in $T_j(h-1)$. In particular, this means that the $h^{th}$ term in the sequence $s$ outputted by $\ell_{\sigma}$ must correspond to edge $\{w_j,c\}$ (since, otherwise, the path corresponding to this sequence would not contain $c$). Hence, the $h^{th}$ term in sequence $s$ is equal to $j$. We showed above that $\ell_{\sigma}$ in $G_{\sigma}$  and $\ell_{\sigma'}$ in $G_{\sigma'}$ output the same port sequence, so the $h^{th}$ term of $\ell_{\sigma'}$'s output is also equal to $j$. Since $j \neq j'$, it follows that the $h^{th}$ edge on the path corresponding to $\ell_{\sigma'}$'s output does not correspond to edge $\{w_{j'},c\}$. Hence, in $G_{\sigma'}$, node $\ell_{\sigma'}$ elects a node in $T_{j'}(h-1)$. This elected node is at distance $|s| - (h-1)$ from $w_{j'}$ since the first $h-1$ ports in $s$ correspond to the path from $\ell_{\sigma'}$ to $w_{j'}$. Since $d(q_0, w_{j'})=h$, we have that, in $G_{\sigma'}$, the length of $q_0$'s output is $h+[|s| - (h-1)] = |s|+1$.
On the other hand, in $G_{\sigma}$, the node elected by $\ell_{\sigma}$ is not in $T_j(h-1)$, which means that the elected node is at distance $|s|-h$ from $c$ (since the first $h$ ports in $s$ correspond to the path from $\ell_{\sigma}$ to $c$.) It follows that, in $G_{\sigma}$, the length of $q_0$'s output is at most $h-1 + [|s|-h] = |s|-1$, i.e.,
shorter than its output in $G_{\sigma'}$. This contradicts the fact that $q_0$ outputs the same sequence in the execution of $ELECT$ in both $G_{\sigma}$ and $G_{\sigma'}$. Therefore, our assumption that $ELECT$ correctly outputs a leader in both $G_{\sigma}$ and $G_{\sigma'}$ was false, so $ELECT$ fails on at least one tree in $\mathcal{T}$. 

This contradicts the correctness of $ELECT$, so our assumption that the same advice is provided for $G_{\sigma}$ and $G_{\sigma'}$ must be wrong. This concludes the proof of the claim.

By Claim \ref{different} there are $|\mathcal{T}| = (\Delta-1)!$ different advice strings, and hence the size of advice is at least
$\log  ( (\Delta-1)!) \in \Omega( \Delta \log \Delta)$. Since $\Delta = \lfloor (\gamma n)^{2/D} \rfloor$, we get a lower bound $\Omega(n^{2/D}\log n)$
on the size of advice.

The following results demonstrate that leader election is solvable in each $G_{\sigma}$ in time $D-3$
(given its map), and that the number of nodes in $G_{\sigma}$ is in $\Theta(n)$.

\begin{claim}\label{Gsolvable}
$\xi(G_{\sigma}) \leq D-3$
\end{claim}
To prove the claim, we show that every node can identify itself in a map of $G_{\sigma}$, and thus can elect the central node $c$. First note that $c$ is at distance at most $h$ from every node in $G_{\sigma}$, so $V_{G_{\sigma}}(c,D-3) = V_{G_{\sigma}}(c,h)$ is equal to $G_{\sigma}$. Therefore, in time $D-3$, node $c$ can identify itself in a map of $G_{\sigma}$. For any $v \neq c$ in $G_{\sigma}$, node $v$ is located in $T_i(h-1)$ for some $i \in \{0,\ldots,\Delta-1\}$. Since $T_i(h-1)$ has height $h-1$, the distance between any two nodes in $T_i(h-1)$ is at most $D-2$. It follows that $V_{G_{\sigma}}(v,D-3)$ contains every node in $T_i(h-1)$ except possibly some leaves. In particular, $V_{G_{\sigma}}(v,D-3)$ contains the path between nodes $w_i$ and $q_i$. This path is the only induced subtree of $T_i(h-1)$ consisting of a path of length $h-2$. It follows that $v$ can identify this path in $V_{G_{\sigma}}(v,D-3)$. Hence, it can identify its endpoints
and it sees their degrees. Therefore $v$ can identify $q_i$, as it is the only node in $G_{\sigma}$ of degree $\Delta +\sigma(i)+2$. 
It follows that $v$ can identify itself in a map of $G_{\sigma}$. This concludes the proof of the claim.

\begin{claim}\label{sizeG}
The number of nodes in $G_{\sigma}$ is in $\Theta(n)$.
\end{claim}
To prove the claim, note that the size of $G_{\sigma}$ is equal to $1 + (3\Delta ^2 + \Delta)/2 + \Delta \cdot |T_i(h-1)|$. We first prove that $|T_i(x)| \leq 3\Delta^x$ by induction on $x$. In the base cases, $|T_i(0)| = 1$ and $|T_i(1)| = \Delta+1$, so $|T_i(x)| \leq 3\Delta^x$ in these two cases. Now, suppose that $x \geq 2$, and that, for all $y < x$, we have $|T_i(y)| \leq 3\Delta^y$. From the construction of $T_i(x)$, we observe that the size of $T_i(x)$ is $1 + \Delta + (x-2) + (\Delta-1)^2\cdot |T_i(x-2)|$. By the induction hypothesis, this is at most $1 + \Delta + (x-2) + (\Delta-1)^2(3\Delta^{x-2}) = 1 + \Delta + (x-2) + [3\Delta^x - 6\Delta^{x-1} + 3\Delta^{x-2}] \leq 1+\Delta + (x-2) + 3\Delta^x - 3\Delta^{x-1}$. Since each of the first three terms is bounded above by $\Delta^{x-1}$, we have shown that $T_i(x)$ has size at most $3\Delta^x$. Therefore, the size of $G_{\sigma}$ is at most $1 + (3\Delta ^2 + \Delta)/2 + 3\Delta^h = 6(( \gamma n)^{1/h})^{h} \leq 6 \gamma n$. This completes the proof of the claim.

\end{proof}

For time $\alpha \cdot diam$ for any constant $\alpha <1/2$, we showed bounds with tight order of magnitude of $\Theta (n)$, except for diameter
$diam \in O(\log ^2n)$. This yields three open questions. 

{\bf P2.} Find close upper and lower bounds on the minimum size of advice for time very close to $diam/2$, i.e., $diam/2 \pm o(diam)$.

{\bf P3.} Find close upper and lower bounds on the minimum size of advice for time below half of the diameter when the diameter is in $O(\log^2 n)$.

{\bf P4.} Find close upper and lower bounds on the minimum size of advice when time is very small, e.g., logarithmic in $n$.

This last problem has an intriguing twist.
At first glance it would seem that the answer to it,
at least  for diameter in $\omega (\log ^2n)$, is $\Theta(n)$. Indeed, the upper bound $O(n)$ on the size of advice holds in this case as well (with the same proof as in Section 6),
and the lower bound $\Omega (n)$ proved for time $\alpha \cdot diam$ for any constant $\alpha <1/2$  should be ``even more true'':
since decreasing the allocated time makes the task more difficult, the required amount of advice should not decrease.
Perhaps surprisingly, this argument overlooks 
the following subtlety. 
We should recall that, for a given time $\tau$, we seek solutions of our minimum advice problem only for trees $T$ with $\xi(T) \leq \tau$ because, for other 
trees, leader election is infeasible in time $\tau$ with any amount of advice. However, for small values of $\tau$, the restriction $\xi(T) \leq \tau$ could sometimes leave
so few trees under consideration that more efficient advice than for larger values of $\tau$ is sufficient. Is this really the case?



\bibliographystyle{plain}

\begin{thebibliography}{12}

\bibitem{AKM01}
S.~Abiteboul, H.~Kaplan, T.~Milo, Compact labeling schemes for ancestor
queries, Proc. 12th Annual ACM-SIAM Symposium on Discrete
Algorithms (SODA 2001), 547--556.

\bibitem{AHU}
A.V. Aho, J.E. Hopcroft, J.D. Ullman, Data Structures and Algorithms, Addison-Wesley 1983. 


\bibitem{An}
D.~Angluin, Local and global properties in networks of processors. 
Proc. 12th Annual ACM Symposium on Theory of Computing (STOC 1980), 82--93.


\bibitem{AtSn}
H. Attiya and M. Snir,
Better Computing on the Anonymous Ring,
Journal of Algorithms 12, (1991), 204-238.



\bibitem{ASW}
H. Attiya, M. Snir, and M. Warmuth,
Computing on an Anonymous Ring,
Journal of the ACM 35, (1988), 845-875.

\bibitem{BSVCGS}
P. Boldi, S. Shammah, S. Vigna, B. Codenotti, P. Gemmell, and J. Simon,
Symmetry Breaking in Anonymous Networks: Characterizations. 
Proc. 4th Israel Symposium on Theory of Computing and Systems,
(ISTCS 1996), 16-26.




\bibitem{BV}
P. Boldi and S. Vigna,
Computing Anonymously with Arbitrary Knowledge,
Proc. 18th ACM Symp. on Principles of Distributed Computing (PODC 1999), 181-188.

\bibitem{B}
J.E. Burns, A Formal Model for Message Passing Systems,
Tech. Report TR-91, Computer Science Department,
Indiana University, Bloomington, September 1980.





%
%

\bibitem{Ch}
F. Chierichetti, personal communication.



\bibitem{DP}
D. Dereniowski, A. Pelc, Drawing maps with advice,  Journal of Parallel and Distributed Computing 72 (2012), 132--143. 

\bibitem{DP1}
D. Dereniowski, A. Pelc, Leader election for anonymous asynchronous agents in arbitrary networks, Distributed Computing 27 (2014), 21-38. 


\bibitem{DoPe}
S. Dobrev and A. Pelc, 
Leader Election in Rings with Nonunique Labels, Fundamenta Informaticae 59 (2004), 333-347. 




\bibitem{EFKR}
Y. Emek, P. Fraigniaud, A. Korman, A. Rosen, Online computation with advice, Theoretical Computer Science 412 (2011), 2642--2656.


\bibitem{FKKLS}
P. Flocchini, E. Kranakis, D. Krizanc, F.L. Luccio and  N. Santoro,
Sorting and Election in Anonymous Asynchronous Rings,
Journal of Parallel and Distributed Computing 64 (2004), 254-265.




\bibitem{FGIP}
P. Fraigniaud, C. Gavoille, D. Ilcinkas, A. Pelc, 
Distributed computing with advice: Information sensitivity of graph coloring, 
Distributed Computing 21 (2009), 395--403.

\bibitem{FIP1}
P. Fraigniaud, D. Ilcinkas, A. Pelc, 
Communication algorithms with advice, Journal of  Computer and System Sciences 76 (2010), 222--232.

\bibitem{FIP2}
P. Fraigniaud, D. Ilcinkas, A. Pelc, 
Tree exploration with advice, Information and Computation 206 (2008), 1276--1287.

\bibitem{FKL}
P. Fraigniaud, A. Korman, E. Lebhar,
Local MST computation with short advice,
Theory of Computing Systems 47 (2010), 920--933.

\bibitem{FL}
G.N. Fredrickson and N.A. Lynch,
Electing a Leader in a Synchronous Ring,
Journal of the ACM 34 (1987), 98-115.

\bibitem{FP2}
E. Fusco, A. Pelc, How Much Memory is Needed for Leader Election, Distributed Computing 24 (2011), 65-78. 

\bibitem{FP1}
E. Fusco, A. Pelc, Knowledge, level of symmetry, and time of leader election, Proc. 20th Annual European Symposium on Algorithms (ESA 2012),  LNCS 7501, 479-490. 






\bibitem{FP}
E. Fusco, A. Pelc, Trade-offs between the size of advice and broadcasting time in trees, Algorithmica 60 (2011), 719--734. 


\bibitem{FPR}
E. Fusco, A. Pelc, R. Petreschi, Use knowledge to learn faster: Topology recognition with advice, Proc. 27th International Symposium on Distributed Computing (DISC 2013), 31-45.

\bibitem{GPPR02}
C.~Gavoille, D.~Peleg, S.~P\'{e}rennes, R.~Raz.
Distance labeling in graphs, 
Journal of Algorithms 53 (2004), 85-112.

   \bibitem{HKMMJ}
     M.A. Haddar, A.H. Kacem, Y. M\'{e}tivier, M. Mosbah, and M. Jmaiel,  Electing a Leader in the Local Computation Model using Mobile Agents.
Proc.  6th ACS/IEEE International Conference on Computer Systems and Applications (AICCSA 2008), 473-480.

\bibitem{HS}
D.S. Hirschberg, and J.B. Sinclair,
Decentralized Extrema-Finding in Circular Configurations of Processes,
Communications of the ACM 23 (1980), 627-628.





\bibitem{IKP}
D. Ilcinkas, D. Kowalski, A. Pelc, 
Fast radio broadcasting with advice, 
 Theoretical Computer Science, 411 (2012),  1544--1557.
 
 \bibitem{JKZ}
T. Jurdzinski, M. Kutylowski, and J. Zatopianski, 
Efficient Algorithms for Leader Election in~Radio Networks.
 Proc., 21st ACM Symp. on Principles of Distributed Computing
(PODC 2002), 51-57.




\bibitem{KKKP02}
M.~Katz, N.~Katz, A.~Korman, D.~Peleg, Labeling schemes for flow and
connectivity, 
SIAM Journal of  Computing 34 (2004), 23--40.


\bibitem{KKP05}
A. Korman, S. Kutten, D. Peleg, Proof labeling schemes,
Distributed Computing 22 (2010), 215--233.  

\bibitem{KP}
D. Kowalski, and A. Pelc, Leader Election in Ad Hoc Radio Networks: A Keen Ear Helps, 
Proc. 36th International Colloquium on Automata, Languages and Programming (ICALP 2009), 
LNCS 5556, 521-533. 



\bibitem{LL}
G. Le Lann,
Distributed Systems - Towards a Formal Approach,
Proc. IFIP Congress, 1977, 155--160, North Holland.




\bibitem{Ly}
N.L. Lynch, Distributed algorithms, Morgan Kaufmann Publ. Inc.,
San Francisco, USA, 1996.

\bibitem{MP}
A. Miller, A. Pelc: Election vs. Selection: Two Ways of Finding the Largest Node in a Graph,
CoRR abs/1411.1319 (2014).

\bibitem{NO}
K. Nakano and S. Olariu, Uniform Leader Election Protocols for Radio Networks,
IEEE Transactions on Parallel and Distributed Systems 13
(2002), 516-526.

\bibitem{SN}
N. Nisse, D. Soguet, Graph searching with advice,
Theoretical Computer Science 410 (2009), 1307--1318.



 \bibitem{Pe}D. Peleg,
  Distributed Computing, A Locality-Sensitive Approach,
  SIAM Monographs on Discrete Mathematics and Applications, Philadelphia 2000.
  
  \bibitem{P}
G.L. Peterson, An $O(n \log n)$ Unidirectional Distributed Algorithm
for the Circular Extrema Problem,
ACM Transactions on Programming Languages and Systems 4 (1982), 758-762.








\bibitem{TZ05}
M.~Thorup, U.~Zwick, Approximate distance oracles,
Journal of the ACM, 52 (2005), 1--24.

\bibitem{Wil}
D.E. Willard, 
Log-logarithmic Selection Resolution Protocols in a Multiple Access Channel,
SIAM J. on Computing 15 (1986), 468-477. 


\bibitem{YK2}
M. Yamashita and T. Kameda,
Electing a Leader when Procesor Identity Numbers are not Distinct,
Proc. 3rd Workshop on Distributed Algorithms (WDAG 1989),
LNCS 392, 303-314.

\bibitem{YK3}
M. Yamashita and T. Kameda,
Computing on Anonymous Networks: Part I - Characterizing the Solvable Cases,
 IEEE Trans. Parallel and Distributed Systems 7 (1996), 69-89. 








\end{thebibliography}


\end{document}